%% file: paper.tex
\documentclass{vldb}
\usepackage{etex}
\reserveinserts{28}
\usepackage{balance}  

\usepackage{amsmath}
\usepackage[makeroom]{cancel}
\usepackage[normalem]{ulem}
\usepackage{soul}
\usepackage{txfonts}
\usepackage{color,listings,booktabs,graphicx}
\usepackage{soul}
\usepackage{subcaption}
\DeclareCaptionType{copyrightbox}
\usepackage{arydshln} 
\usepackage{multirow}
\usepackage{pgfplots}
\usepackage{filecontents}
\usepackage{xspace}
\usepackage{enumitem} 
\usepackage{hyperref}

\usepackage{algorithm}
\usepackage{algpseudocode}

\setlist{noitemsep,topsep=2pt,parsep=2pt,partopsep=0pt,leftmargin=*} 
\renewcommand{\paragraph}[1]{\vspace{3pt}{\bf #1.}} 
\newcommand{\designsubsection}[1]{\subsection{#1}}
\newcommand{\figref}[1]{Figure \figrefnum{#1}}
\newcommand{\figrefnum}[1]{\ref{fig:#1}}
\newcommand{\subfigref}[2]{Figure \subfigrefnum{#1}{#2}}
\newcommand{\subfigrefnum}[2]{\figrefnum{#1}(\subref{fig:#2})}

\newcommand{\newProof}[0]{\iffalse}

\iffalse 
	\newcommand{\highlight}[1]{\colorbox{yellow}{#1}}
	\newcommand{\checkOK}[1]{\colorbox{red!30}{#1}}
	\newcommand{\todo}[1]{{\color{red}TODO: {#1}}}
	\newcommand{\question}[1]{{\color{magenta}Question: {#1}}}
	
	\newcommand{\idea}[1]{{\color{magenta}IDEA: {#1}}}
	\newcommand{\CK}[1]{{\color{magenta}[CK: {#1}]}}
	\newcommand{\MD}[1]{{\color{darkgreen}[MD: {#1}]}}
	\newcommand{\SBJ}[1]{{\color{cyan}[SBJ: {#1}]}}
	\newcommand{\AS}[1]{{\color{magenta}[AS: {#1}]}}
	\newcommand{\FIXED}[1]{{\color{gray}[FIXED: {#1}]}}
	\newcommand{\forgetNow}[1]{}

	\iffalse 
		\newtheorem{definition}{\colorbox{magenta!30}{Definition}}[section]
		\newtheorem{theorem}{\colorbox{yellow!30}{Theorem}}[section]
		\newtheorem{proposition}{\colorbox{blue!30}{Proposition}}[section]
		\newtheorem{corollary}{\colorbox{red!30}{Corollary}}[section]
		\newtheorem{lemma}{\colorbox{green!30}{Lemma}}[section]
		\newtheorem{theexample}{\colorbox{brown!30}{Example}}
	\else
		\newtheorem{definition}{Definition}[section]
		\newtheorem{theorem}{Theorem}[section]
		\newtheorem{proposition}{Proposition}[section]
		\newtheorem{corollary}{Corollary}[section]
		\newtheorem{lemma}{Lemma}[section]
		\newtheorem{theexample}{Example}
	\fi

\else 
	\newcommand{\checkOK}[1]{#1}
	\newcommand{\highlight}[1]{#1}
	\newcommand{\todo}[1]{}
	\newcommand{\question}[1]{}

	\newcommand{\idea}[1]{}
	\newcommand{\CK}[1]{}
	\newcommand{\MD}[1]{}
	\newcommand{\SBJ}[1]{}
	\newcommand{\AS}[1]{}
	\newcommand{\FIXED}[1]{}
	\newcommand{\forgetNow}[1]{}

	\newtheorem{definition}{Definition}[section]
	\newtheorem{theorem}{Theorem}[section]
	\newtheorem{proposition}{Proposition}[section]
	\newtheorem{corollary}{Corollary}[section]
	\newtheorem{lemma}{Lemma}[section]
	\newtheorem{theexample}{Example}
\fi

\newcommand\xqed[1]{%
  \leavevmode\unskip\penalty9999 \hbox{}\nobreak\hfill
  \quad\hbox{#1}}
\newcommand\exend{\xqed{{\large $\triangle$}}}

\newenvironment{example}{
   \begin{theexample} \bgroup \em }{
   \egroup \exend \end{theexample}}
\newenvironment{contexample}{
   \addtocounter{theexample}{-1} \begin{theexample}[continued] \bgroup \em}{
   \egroup \exend \end{theexample}}

\newenvironment{proofsketch}
{
\trivlist
    \item[%
        \hskip 10pt
        \hskip \labelsep
        {\sc Proof Sketch.}%
    ]
    \ignorespaces
}{
\if@qeded\else\qed\fi
    \endtrivlist
}

\newcommand{\pheader}[1]{\textbf{#1.}}

\newcommand{\fix}[1]{{\color{red}}}

\definecolor{darkgreen}{RGB}{0,128,0}
\definecolor{dkpink}{RGB}{200,0,100}
\definecolor{gray}{RGB}{128,128,128}
\definecolor{mygray}{rgb}{0.5,0.5,0.5}

\usepackage{amssymb}
\usepackage{pifont}
\newcommand{\comm}[1]{}

\def\systemFullName{Multi-Version Concurrency Control with Closures\xspace}
\def\NMVCCInitials{O\xspace}
\def\NMVCC{OMVCC\xspace}
\def\Validation{Validation\xspace}
\def\Repair{Repair\xspace}
\def\systemName{MV3C\xspace}
\def\MVCC{MVCC\xspace}

\makeatletter
\newenvironment{btHighlight}[1][]
{\begingroup\tikzset{bt@Highlight@par/.style={#1}}\begin{lrbox}{\@tempboxa}}
{\end{lrbox}\bt@HL@box[bt@Highlight@par]{\@tempboxa}\endgroup}

\newcommand\btHL[1][]{%
  \begin{btHighlight}[#1]\bgroup\aftergroup\bt@HL@endenv%
}
\def\bt@HL@endenv{%
  \end{btHighlight}%
  \egroup
}
\newcommand{\bt@HL@box}[2][]{%
  \tikz[#1]{%
    \pgfpathrectangle{\pgfpoint{1pt}{0pt}}{\pgfpoint{\wd #2}{\ht #2}}%
    \pgfusepath{use as bounding box}%
    \node[anchor=base west, fill=orange!30,outer sep=0pt,inner xsep=1pt, inner ysep=0pt, rounded corners=3pt, minimum height=\ht\strutbox+1pt,#1]{\raisebox{1pt}{\strut}\strut\usebox{#2}};
  }%
}
\makeatother

\lstdefinelanguage{PLSQL}%
{morekeywords={IF,START,BEGIN,ABSOLUTE,ACTION,ADD,ALLOCATE,ALTER,ARE,AS,ASSERTION,%
AT,BETWEEN,BIT_LENGTH,BOTH,BY,CASCADE,CASCADED,CASE,CAST,%
CATALOG,CHAR_LENGTH,CHARACTER_LENGTH,CLUSTER,COALESCE,%
COLLATE,COLLATION,COLUMN,CONNECT,CONNECTION,CONSTRAINT,%
CONSTRAINTS,CONVERT,CORRESPONDING,CREATE,CROSS,CURRENT_DATE,%
CURRENT_TIME,CURRENT_TIMESTAMP,CURRENT_USER,DAY,DEALLOCATE,%
DEC,DEFERRABLE,DEFERED,DESCRIBE,DESCRIPTOR,DIAGNOSTICS,%
DISCONNECT,DOMAIN,DROP,ELSE,END,EXEC,EXCEPT,EXCEPTION,EXECUTE,%
EXTERNAL,EXTRACT,FALSE,FIRST,FOREIGN,FROM,FULL,GET,GLOBAL,%
GRAPHIC,HAVING,HOUR,IDENTITY,IMMEDIATE,INDEX,INITIALLY,INNER,%
INPUT,INSENSITIVE,INSERT,INTO,INTERSECT,INTERVAL,%
ISOLATION,JOIN,KEY,LAST,LEADING,LEFT,LEVEL,LIMIT,LOCAL,LOWER,%
MATCH,MINUTE,MONTH,NAMES,NATIONAL,NATURAL,NCHAR,NEXT,NO,NOT,NULL,%
NULLIF,OCTET_LENGTH,ON,ONLY,ORDER,ORDERED,OUTER,OUTPUT,OVERLAPS,%
PAD,PARTIAL,POSITION,PREPARE,PRESERVE,PRIMARY,PRIOR,READ,%
RELATIVE,RESTRICT,REVOKE,RIGHT,ROWS,SCROLL,SECOND,SELECT,SESSION,%
SESSION_USER,SIZE,SPACE,SQLSTATE,SUBSTRING,SYSTEM_USER,%
TABLE,TEMPORARY,THEN,TIMEZONE_HOUR,%
TIMEZONE_MINUTE,TRAILING,TRANSACTION,TRANSLATE,TRANSLATION,TRIM,%
TRUE,UNIQUE,UNKNOWN,UPPER,USAGE,USING,VALUE,VALUES,%
VARGRAPHIC,VARYING,WHEN,WHERE,WRITE,YEAR,ZONE,%
AND,ASC,avg,CHECK,COMMIT,count,DECODE,DESC,DISTINCT,GROUP,IN,
LIKE,NUMBER,ROLLBACK,SUBSTR,sum,VARCHAR2,
MIN,MAX,UNION,UPDATE,
ALL,ANY,CUBE,CUBE,DEFAULT,DELETE,EXISTS,GRANT,OR,RECURSIVE,
ROLE,ROLLUP,SET,SOME,TRIGGER,VIEW},
morendkeywords={BIT,BLOB,CHAR,CHARACTER,CLOB,DATE,DECIMAL,FLOAT,
INT,INTEGER,NUMERIC,SMALLINT,TIME,TIMESTAMP,VARCHAR},
sensitive=false,
morecomment=[l]--,%
morecomment=[s]{/*}{*/},%
morestring=[d]',%
morestring=[d]",%
moredelim=**[is][\btHL]{`}{`},
moredelim=**[is][{\btHL[fill=green!30,draw=red,dashed,thin]}]{@}{@},
numbers=left,                    
numbersep=2pt,                   
numberstyle=\tiny\color{mygray}, 
}[keywords,comments,strings]%

\lstdefinelanguage{Scala}{
	keywords={abstract,case,catch,class,def,do,else,extends,false,final,finally,for,if,implicit,import,%
	match,mixin,new,null,object,override,package,lazy,private,protected,requires,return,sealed,super,this,%
	throw,trait,true,try,type,val,var,while,with,yield},
	otherkeywords={=>,<-,<\%,<:,>:,\#,@},sensitive=true,
	morecomment=[l]{//},	morecomment=[n]{/*}{*/},
	morestring=[b]",morestring=[b]',morestring=[b]""",showstringspaces=false
}

\lstdefinelanguage{ScalaXact}{
	keywords={abstract,case,catch,class,def,do,else,extends,false,final,finally,for,if,implicit,import,%
	match,mixin,new,null,object,override,package,lazy,private,protected,requires,return,sealed,super,this,%
	throw,trait,true,try,type,val,var,while,with,yield,SELECT,FROM,WHERE,EXEC,ORDER,BY,DELETE,SQL,COMMIT,WORK},
	otherkeywords={=>,<-,<\%,<:,>:,\#,@},sensitive=true,
	morecomment=[l]{//},	morecomment=[n]{/*}{*/},
	morestring=[b]",morestring=[b]',morestring=[b]""",showstringspaces=false
}

\newcommand\concat{+}

\begin{document}

\title{Repairing Conflicts among MVCC Transactions}

\numberofauthors{1}
\author{
\alignauthor 
Mohammad Dashti,
Sachin Basil John, 
Amir Shaikhha,
and Christoph Koch\\
\hspace{1cm}\\
\affaddr{DATA, \'Ecole Polytechnique F\'ed\'erale de Lausanne (EPFL), Switzerland \hspace{0.1cm} \{firstname\}.\{lastname\}@epfl.ch}
}
\date{30 January 2016}

\maketitle

\begin{abstract}
The optimistic variants of MVCC (Multi-Version Concurrency Control) avoid blocking concurrent transactions at the cost of having a validation phase. Upon failure in the validation phase, the transaction is usually aborted and restarted from scratch. The ``abort and restart'' approach becomes a performance bottleneck for the use cases with high contention objects or long running transactions. In addition, restarting from scratch creates a negative feedback loop in the system, because the system incurs additional overhead that may create even further conflicts.

In this paper, we propose a novel approach for conflict resolution in MVCC for in-memory databases. This low overhead approach summarizes the transaction programs in the form of a dependency graph. The dependency graph also contains the constructs used in the validation phase of the MVCC algorithm. Then, in the case of encountering conflicts among transactions, the conflict locations in the program are quickly detected, and the conflicting transactions are partially re-executed. This approach maximizes the reuse of the computations done in the initial execution round, and increases the transaction processing throughput.
\end{abstract}

\input{intro}
\input{relatedwork}
\input{design}
\input{theory}
\input{interop}
\input{optimizations}
\input{impl}
\input{eval}

\bibliographystyle{abbrv}
\small
\bibliography{refs}

\end{document}

%% file: intro.tex
\section{Introduction} \label{sec:intro}

Recent research proposes an optimistic \MVCC algorithm as the best fit for concurrency control in in-memory databases \cite{neumann15}. This algorithm, like its predecessors \cite{larson11,jones10}, takes the path of aborting and restarting the conflicting transactions, which is simple but sub-optimal. 
Any conflict among transactions results in more work for the concurrency control algorithm and an increase in the execution latency of the transactions. The increased latency of the transaction execution not only affects the throughput of individual transactions, but also increases the probability of having more concurrent transactions in the future. This might incur even more conflicts, forming a negative feedback loop.

The sub-optimality of the abort and restart approach becomes more significant when the number of transactions that are conflicting is high. The two factors that contribute the most to an increase in the number of conflicts are: (1) having high contention data objects that are read and updated by several concurrent transactions, and (2) having long running transactions, the lifespan of which intersects with that of many other transactions.

For the aforementioned scenarios, one could propose pessimistic concurrency control algorithms. However, it should be noted that, firstly, a pessimistic approach to concurrency control yields low performance for long running transactions, as the acquisition of a highly demanded resource by a long running transaction requires either preventing the long-running transaction from committing, or stopping almost all other transactions. Secondly, general DBMSs normally implement a specific general purpose concurrency control algorithm and are optimized for it. A pessimistic algorithm may not be an option.

In this paper, we introduce a novel multi-version timestamp ordering concurrency control algorithm, referred to as  \systemFullName (\systemName). This algorithm resolves the conflicts among concurrent transactions by only partially aborting and restarting them. The main challenges for a conflict resolution technique, such as the one in \systemName, are having: (1) a low overhead on the normal execution of transactions, as this overhead cost is paid for each transaction, regardless of encountering a conflict or not, and (2) a fast mechanism to narrow down the conflicting portions of the transactions and fixing them, as a mechanism that is slower than the abort and restart approach defeats the purpose.
The first challenge is dealt with by reducing additional book-keeping by reusing the existing concurrency control machinery. 
To address the second challenge, \systemName uses lightweight annotations on the transaction programs that help in identifying the dependencies among different operations. Using these annotations, \systemName pinpoints blocks of the program affected by the conflicts and quickly re-executes only those blocks. The annotations are added either manually by the user, or by employing static program analysis and restructuring.

The rationale for proposing \systemName is that a conflict happens only if a transaction reads some data objects from the database, which become stale by its commit attempt time. Here, the assumption is that all data modifications made by a specific transaction are invisible to the other transactions during its execution. These modifications become visible only after the critical section during which the transaction commits. Consequently, just before committing a transaction, by checking whether its data lookup operations read the most recent (committed) versions of the data objects, serializability of the execution is ensured. In addition, through associating read operations with the blocks of code that depend on them, the portion of the transaction that should be re-executed in the case of a conflict is identified quickly, and gets re-executed. 
These blocks correspond to a specific class of sub-transactions in the nested  transaction model. The boundaries of these blocks are specified by \systemName, which makes it possible to efficiently repair conflicting  transactions. This is discussed further in section \ref{sec:related:compilation}.

In particular, this paper makes the following three contributions:

\begin{enumerate}
\item An efficient conflict resolution mechanism for multi-version \checkOK{data-bases}, \systemName, which repairs conflicts instead of aborting transactions, with a minimum execution overhead. The design of this mechanism is discussed in section \ref{sec:design}.
\item A method to deal with write-write conflicts in \systemName. Unlike other optimistic MVCC algorithms, \systemName can avoid aborting the transaction prematurely when a write-write conflict is detected. This is discussed further in section \ref{sec:writewrite}.
\item A mechanism for fixing the result-set of failed queries for \systemName, which can optionally be enabled for each query in a transaction. This mechanism can boost the repair process for transactions that have long running queries. The details of this mechanism are described in section \ref{sec:reuse}.
\end{enumerate}

\noindent \pheader{Motivating examples}
There are different cases of transaction programs that benefit from \systemName. The three main categories are illustrated in Figure \ref{fig:cases}. It should be noted that a combination of these three cases can create more complicated scenarios, where \systemName is even more effective. Each case in this figure shows an instance of a transaction program starting with a {\em begin} command and finishing with a {\em commit} command. Each program consists of one or more blocks of code represented by a box. The dependencies among different blocks of code are represented by arrows. If block $Y$ depends on block $X$, there is an arrow from $X$ to $Y$. Moreover, it is assumed that each of these three instances failed during their validation phase, because of a conflict detected in its block $A$. 

\begin{figure}[htb]
\begin{center}
\leavevmode
\includegraphics[width=\columnwidth]{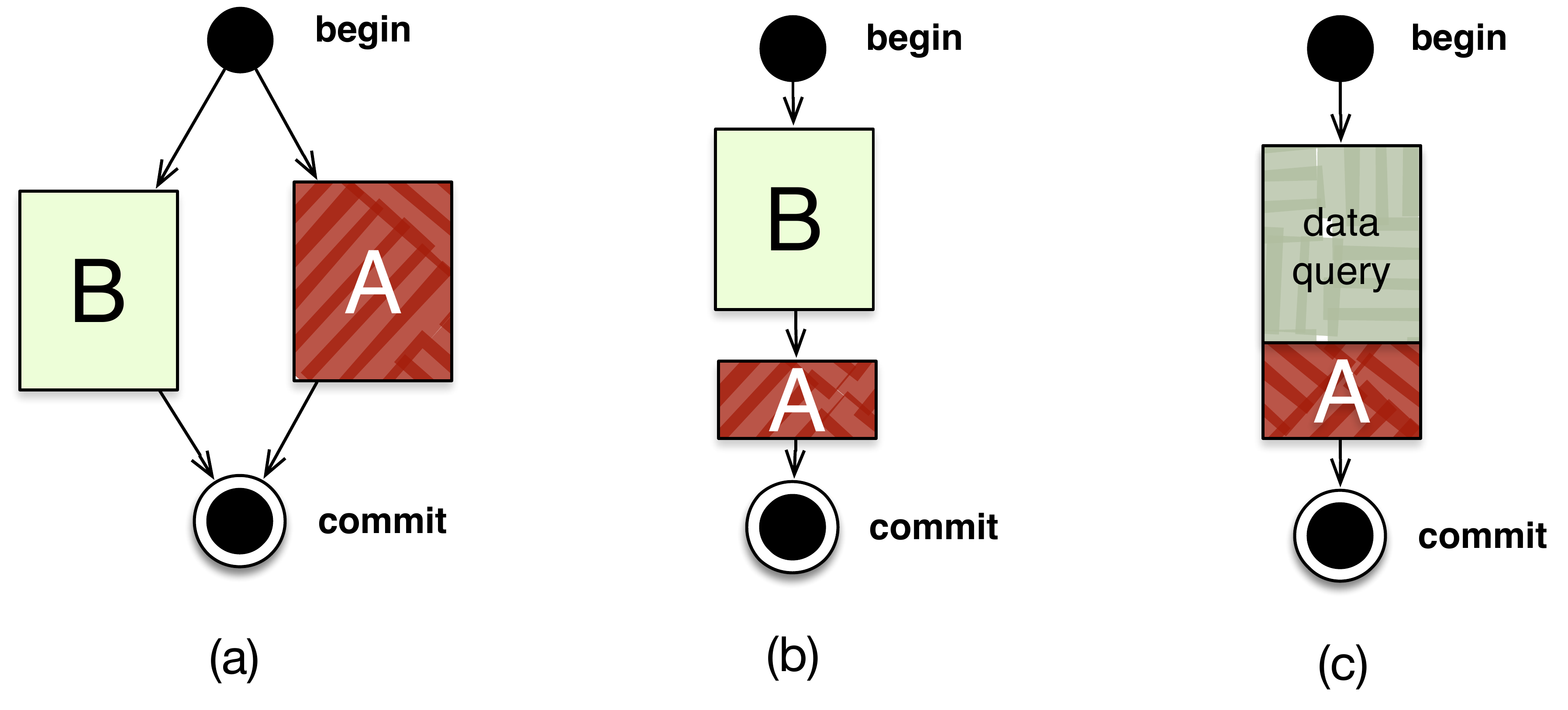}
\end{center}
\vspace{-0.2cm}
\caption{Cases where \systemName is more efficient in repairing the conflicting transactions compared to the ``abort and restart" approach.}
\label{fig:cases}	
\end{figure}

\textbf{First.} In this case, the transaction has logically disjoint program paths and conflicts happen only in a few of them. Such paths are detected from the program structure and only they are re-executed. One example of this case is shown in Figure \ref{fig:cases}(a).

\begin{example}
Assume that in Figure \ref{fig:cases}(a), block $A$  reads a row from table $T_A$ and updates it, and block $B$ reads a row from table $T_B$ and updates it. Moreover, these updates only depend on the input parameters of the program. Then, if a concurrent transaction also updates the same row from table $T_A$ and commits before the other transaction, the latter fails to commit. However, \systemName detects that only block $A$ has a conflict and re-executes only this block, without re-executing block $B$.
\end{example}

\textbf{Second.} In this case, conflicts happen after doing a substantial amount of work in the transaction. Here, it is not necessary to redo all the work. Instead, the data available before the conflict is reused in order to continue from the point of conflict. Figure \ref{fig:cases}(b) shows this case, where only the conflicting block $A$ is re-executed. 
A concrete example of this case is the banking example described in Example \ref{ex:banking}. In the following sections, this example is used in order to better describe \systemName.

\begin{example} \textbf{Banking example:} \label{ex:banking}
The example consists of a simplified banking database with an Account table. This table stores the balance of the customers identified by an ID. There are two types of transaction programs that run on this database. The first program, named SumAll, is read-only. A SumAll transaction sums up the balances in all the existing accounts. The second program is named TransferMoney and it is written in a PL/SQL-like language as shown in Figure \ref{fig:example1}. A TransferMoney transaction transfers a specific amount of money from one account to the other, given the availability of sufficient funds. The money transfer also consists of a fee that is deducted from the sender account and is added to the central fee account identified by $FEE\_ACC\_ID$.

Now, assume that two TransferMoney transactions, using different input parameters, run concurrently. Then, the first one succeeds, and the other one fails due to line \ref{fig:example1:updateFeeLine} in Figure \ref{fig:example1}. \systemName detects that only the operation that reads the current value of the fee account impacts the correctness of line \ref{fig:example1:updateFeeLine}. Thus, only that line gets re-executed, this time, with the new value of the fee account. 
\end{example}

\textbf{Third.} In this case, conflicts occur in the beginning of the transaction. Then, the data returned by SELECT queries to the database is reused after accommodating the changes introduced by the conflicting transaction(s). Thus, the re-evaluation of queries from \checkOK{scrat-ch} is avoided. As illustrated in Figure \ref{fig:cases}(c), even though the transaction program consists of a single block of code, the initial part of the block responsible for querying the data from the database is re-executed more efficiently under \systemName.

\begin{contexample}
In the banking example, assume that there is another transaction program, named Bonus. This transaction program increments the balance of the accounts with a minimum balance of 500 CHF by 1 CHF. As the {\em balance} column is not indexed, a Bonus transaction has to scan the whole Account table. Meanwhile, a concurrent TransferMoney transaction commits, increasing the balance of an account above 500 CHF. Then, the Bonus transaction fails validation, as it did not consider the new record in the Account table. However, as \systemName knows that the conflict happened only because of that record, it fixes the result-set of the query by including the additional record. This completely avoids another round of full scan over the Account table.
\end{contexample}

%% file: relatedwork.tex
\section{Related work} \label{sec:related}

Transaction processing is a fundamental area of database research. Work in this area has resulted in many publications and books \cite{bernstein86, gray92, weikum01}. This section covers work done on topics related to \systemName. such as multi-version concurrency control and nested transactions.

\subsection{Multi-Version Concurrency Control} \label{sec:related:mvcc}
Concurrency control techniques from the MVCC family (which includes snapshot isolation) are de facto standard in open source as well as commercial database management and transaction processing systems. PostgreSQL \cite{ports12}, Microsoft SQL Server's Hekaton  \cite{hekaton}, SAP HANA \cite{sap} and HyPer \cite{kemper11} are among these systems. In addition, there is recent work proposing efficient MVCC algorithms for in-memory transaction processing \cite{larson11, merrifield13, neumann15}.
As it was mentioned in section \ref{sec:intro}, \systemName has an efficient conflict resolution mechanism, which is missing in the existing MVCC algorithms. \systemName builds this mechanism on top of the algorithm proposed in \cite{neumann15}. Throughout the rest of this paper, this algorithm is referred to as \NMVCC, where \NMVCCInitials stands for {\em Optimistic}. The details regarding \NMVCC are presented in section \ref{sec:overview}.

\subsection{Nested transactions} \label{sec:related:compilation}
In the nested transaction model \cite{aspnes88, moss85, fekete90}, a transaction consists of primitive actions or sub-transactions, which can again be nested transactions. In this model, after detecting a conflict inside a sub-transaction, only the sub-transaction is aborted, the state at its start time is recreated, and it is re-executed.
The boundaries of such sub-transactions are user-defined, and checkpointing is the usual technique used for recreating the state.

Irrespective of the operations done in the sub-transaction, checkpointing has to be pessimistic. A checkpoint can be used anywhere within the program and no assumptions about the program can be made during its creation. In addition, the number of extra computation cycles required for checkpointing is not negligible. It requires capturing the whole execution state, and this work is not beneficial to the execution of the main transaction program. Moreover, all work done in a \highlight{conflicting} sub-transaction is lost, and cannot be reused in the subsequent re-execution of the sub-transaction.

\systemName achieves a higher throughput compared to the generic nested transaction model due to some key differences.
\systemName has stricter regulations that do not allow an arbitrary splitting of programs. Instead, the boundaries are defined based on the possible failure points in the program and the dependencies among the operations. Consequently,
a faster and more compact checkpointing is achieved, as \systemName tailors the checkpoint for each sub-transaction specifically, unlike the generic checkpointing used in the nested transaction model.
Also, \systemName does not commit its so called sub-transactions, and only tries to commit the transaction as a whole. Thus, \systemName incurs low overhead for executing transactions while still being able to efficiently repair conflicting transactions. This is possible only because of the  constraints in the definition of boundaries for the so-called sub-transactions in \systemName.

%% file: design.tex
\section{\systemName Design} \label{sec:design}
The main idea behind \systemName is that by annotating the transaction programs and thereby exposing the program dependencies to the concurrency control algorithm, conflicts among concurrent transactions can be resolved efficiently. This information is already provided to the transaction processing system via user defined transaction programs \cite{thomson12, yan16}. However, an abstract view of the program is needed to exploit this dependency information. In this abstract view, correctness checks are associated with blocks of code. Then, each executed instance of the transaction program, referred to as a {\em transaction}, uses this information in order to recover efficiently from a failure due to a conflict.

For this purpose, the possible failure points are identified and the transaction program is partitioned into smaller fragments such that each \highlight{failure} is contained within a single fragment. These program fragments are independent, and the failure of some fragments do not affect the other fragments in any way. The failure possibility stems from having a predicate in that fragment of the program, which might not pass the validation phase. After a failed validation, a new timestamp is assigned to the transaction, but only those \highlight{fragments} that have an invalid predicate inside them are rolled back and re-executed. The validation semantics guarantees that the other predicates would return the same values as the initial execution, and therefore re-executing them is unnecessary.

In the rest of this section, we describe the design of \systemName. A brief overview of \NMVCC is provided before going into the details of \systemName, as the latter borrows some features from the former.

\begin{figure}
\lstset{escapeinside={(*@}{@*)}}
\begin{lstlisting}[language=PLSQL]
/* fm = from, acc = account and bal = balance */
TransferMoney(fm_acc, to_acc, amount) {
 START;
 SELECT bal INTO :fm_bal FROM `Account WHERE id=:fm_acc`;
 
 IF(amount < 100) fee = 1.0;
 ELSE fee = amount * 0.01;
 
 IF(fm_bal > amount+fee) {
	SELECT bal INTO :to_bal FROM `Account WHERE id=:to_acc`;

	fm_bal_final = fm_bal - (amount + fee);
	to_bal_final = to_bal + amount;

	UPDATE `Account` SET bal=:fm_bal_final `WHERE id=:fm_acc`;
	UPDATE `Account` SET bal=:to_bal_final `WHERE id=:to_acc`;
	UPDATE `Account` SET bal=bal+:fee `WHERE id=:FEE_ACC_ID`; (*@\label{fig:example1:updateFeeLine}@*)
  COMMIT;
 } ELSE ROLLBACK;
}
\end{lstlisting}
\caption{TransferMoney transaction program from the banking example (Example \ref{ex:banking}) in a PL/SQL-like language.}
\label{fig:example1}
\end{figure}

\designsubsection{\NMVCC overview} \label{sec:overview}
\NMVCC gathers only predicates for all the read operations of a transaction. A predicate in \NMVCC can be thought of as a logical predicate created using the attributes of a relation, encapsulating a data selection criterion. For example, the WHERE clause in a SELECT statement on a single table is a predicate for that table. The candidate predicates in our example program are highlighted in Figure \ref{fig:example1}. Moreover, \NMVCC makes a reasonable assumption that every transaction writes into only a limited number of data objects. Therefore, it is feasible to keep track of the write-set of a transaction (i.e., the {\em undo buffer} in \NMVCC) during its execution. While updating the value of a data object, a new version is created for it by the transaction. The same notion of {\em version} is used in \systemName, and it is defined below.

\begin{definition} \label{def:dv}
A \textbf{version} is a quadruple $(T,O,A,N)$, where $T$ is the commit timestamp of the transaction that created the version, or the transaction ID if the transaction is not committed, $O$ is the identifier for the associated data object, $A$ is the value of $O$ maintained in this version, and $N$ is a version identifier if more than one version is written by $T$ for $O$.
\end{definition}

Each data object keeps a list of versions belonging to it, called its {\em version chain}. When a version is created for a data object, it is added to the head of  this chain. In addition, as an optimization for data storage and retrieval, the value in the new version is written directly into the data object itself, and the old value is stored in the newly created version. This version storage technique is used in \NMVCC in order to improve the scan performance, as it avoids some pointer chasing. Definition \ref{def:dv} holds even with this storage optimization, though the value for each version is not physically stored in the version itself, but in a fixed and deterministic offset from it. If a version $V$ is in the head of version chain, its actual value is in the data object itself. Otherwise, the newer version of $V$ is holding the actual value that belongs to $V$.

The value written in a version is immutable and cannot be modified, even by its owner transaction. Thus, given a version $V$, the version identifier $N$ in $V$ is used to differentiate distinct versions written by the same transaction for the same data object. However, after a transaction gets committed, only the newest version written by the transaction becomes visible to the other transactions. In practice, $N$ can be a pair of pointers, pointing to the older and newer versions in the internal chain of versions written by a single transaction. Then, the committed version is the one without a newer version in $N$. The notion of a committed version is defined below.

\begin{definition} \label{def:cdv}
A \textbf{committed version} is a triple $(T,O,A)$, where $(T,O,A,N)$ is a version such that N is the identifier of the latest version written for $O$ by the transaction with commit timestamp $T$.
\end{definition}

The read operations in \NMVCC return the value inside the data object itself if the version chain is empty. Otherwise, the visible value is reconstructed from the value inside the data object by applying the changes from the version chain, until a visible version is reached. A visible version is either owned by the transaction itself or the latest committed version before the transaction started.

\begin{definition} \label{def:vdv}
\textbf{Visible version:} a version $(T_1,O,A,N_1)$ is visible to a transaction with start timestamp $T_2$ if either:
\begin{itemize}
  \item $T_1$ is committed, $T_1 < T_2$ and there is no other version $(T_3,O,\_,\_)$ where $T_1 < T_3 <T_2$, or
  \item $T_1=T_2$ and there is no other version $(T_1,O,\_,N_2)$ where $N_2$ is newer than $N_1$.
\end{itemize}
\end{definition}

\systemName transactions, like those in \NMVCC, have an undo buffer that maintains the list of versions created by the transaction. When a transaction gets committed, its undo buffer contains only the committed versions. As the following proposition suggests, the undo buffer is a representative of the effects of a committed transaction.

\begin{proposition} \label{prp:undobuffer}
The only effects of a committed transaction that are visible to the other transactions are the committed versions in its undo buffer.
\end{proposition}

\begin{proofsketch} 
The above proposition is indirectly proved in \cite{neumann15}. Every transaction keeps a reference to the versions it created in its undo buffer. The proposition follows immediately as the committed state of a multi-version database consists of the committed versions created by the transactions executed on the database.
\end{proofsketch}

Moreover, to capture the changes done to the database between two timestamps, it is sufficient to look at all the transactions that committed during this time period, as the following corollary shows.

\begin{corollary} \label{cor:databasechanges}
The changes to a multi-version database between two timestamps $S$ and $C$ are represented by the versions in the undo buffer of the transactions committed between $S$ and $C$.
\end{corollary}

As \NMVCC is an optimistic algorithm, it requires a validation phase before a successful commit. A variant of {\em precision locking} \cite{jordan81} is used for validating transactions. This approach requires that the result-sets for all the read operations are still valid at commit attempt time, as if the operations were done at that time. This creates an illusion that the whole transaction is executed at commit time.

In order to achieve this, all committed transactions are stored in a list called {\em recently committed} transactions. During the validation of a transaction, all of its predicates are checked against all committed versions of the transactions in the recently committed list. From this list, only those transactions that committed during the lifetime of the one being validated are considered. Based on Corollary \ref{cor:databasechanges}, the undo buffers of these transactions contain the changes of this time period. If any committed version matches against a predicate, then the data read by the transaction is obsolete. The newer version should have been read instead, if the read operation used the commit timestamp as its reference point.  In this case, the transaction fails validation, is rolled back and restarted.

\designsubsection{\systemName machinery} \label{sec:machinery}
\systemName creates new versions for each modification of the data objects, and gathers the predicates used during the execution. In addition, each predicate has a closure bound to it. The term {\em closure} is used in the sense of the term from the programming languages literature \cite{Landin64}.  In the rest of this paper, the term {\em predicate} refers to an {\em \systemName predicate}, unless otherwise stated.

\begin{definition} \label{def:predicate}
An \textbf{\systemName predicate} $X$ consists of (1) a data selection criterion, (2) a closure bound to it, represented by C(X), (3) a list of versions registered to it, represented by $V(X)$, and (4) a list of child predicates, represented by $D(X)$. $V(X)$ and $D(X)$ are populated after $C(X)$ gets executed.
\end{definition}

\begin{definition} \label{def:closure}
$\mathbf{C(X)}$ (\textbf{the closure bound to predicate $\mathbf{X}$}) is a deterministic function that encloses all operations in a transaction program that depend on the result of $X$. $C(X)$ receives the result of evaluating $X$ along with a set of immutable context variables as its parameters.
\end{definition}

A closure can contain data selection or data manipulation operations apart from computations. Each data selection operation creates a new predicate with its own closure, which has access to its enclosing predicates and their result-sets. Moreover, each data manipulation operation  (i.e., insert, delete and update statements) creates a new version for the data object.

\systemName requires annotating the transaction program in order to expose the relationships of the predicates to the versions and to the other predicates. These annotations can be added to the transaction program automatically using program analysis techniques. It is important to notice that the dependency annotations are non-intrusive constructs, and only a simple dependency analysis is required to generate na\"{i}ve pessimistic annotations on a program written in a PL/SQL-like language. One such na\"{i}ve annotation can be derived by assuming that each operation depends on all its previous operations. Moreover, program restructuring can be used along with a detailed program analysis in order to generate more accurate annotations. However, in this paper, we assume that these annotations are given along with the transaction programs, and we leave the automatic program analysis and restructuring as future work.

\begin{figure}[htb]
\begin{center}
\leavevmode
\includegraphics[width=0.9\columnwidth]{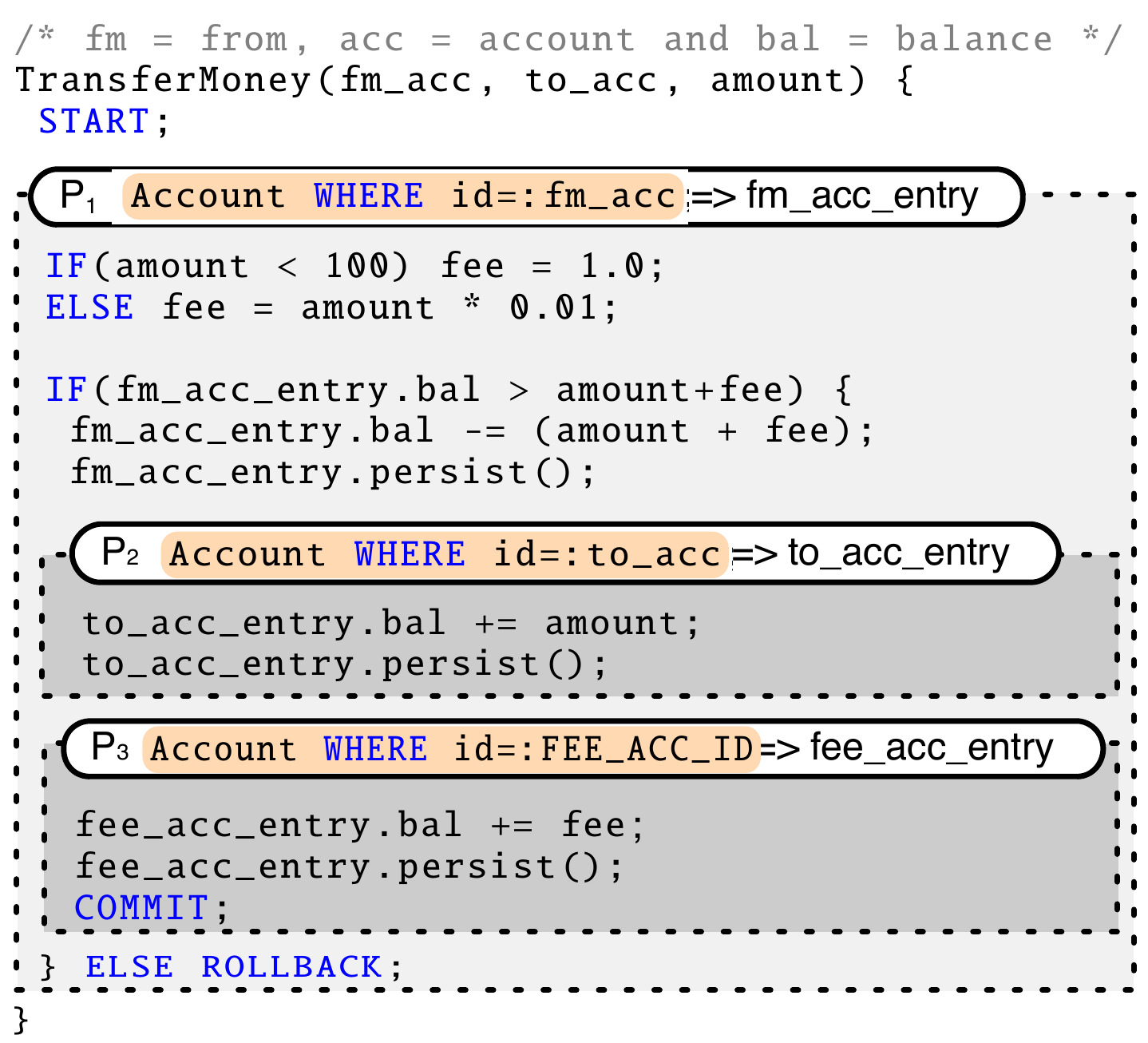}
\end{center}
\vspace{-0.2cm}
\caption{An example transaction program in the \systemName DSL.}
\label{fig:example1_translated}
\end{figure}

To specify the relationships among predicates, when a predicate is created in the closure $C(X)$ of some predicate $X$, it is added to $D(X)$, the list of child predicates of $X$. A graph is constructed internally during the execution, using this parent-child relationship information. By definition, the predicates form a Directed Acyclic Graph (DAG), because there cannot exist any edge from a new predicate to an older predicate and hence a cycle cannot be formed. More intuitively, a SELECT statement cannot get executed if it requires a parameter that is available only after executing another SELECT statement in the future. Consequently, the result of executing a transaction program in \systemName is  a DAG of predicates with a closure assigned to each predicate. The DAG for a transaction $T$ is called the {\em predicate graph of $T$}. In addition, to specify the relationships between predicates and versions, references to the newly created versions are stored into  $V(X)$, the list of versions registered to $X$. Then, using these references, $X$ keeps track of the created versions that directly depend on its result-set.

Apart from the data objects that are altered in the database, there can also be mutable variables in a transaction program. Each mutable variable is local to a closure, as sharing it among more than one closure makes it impossible to reason about the state of the variable, if any of these closures is re-executed. However, an immutable copy of the variables defined in a closure is shared with the closures of its child predicates.

For $C(X)$, the context variables are those variables that are defined outside the closure and are accessible to it. The context variables are immutable and are categorized into: (1) the input parameters of the transaction, (2) the result-sets of the ancestor predicates of $X$, and (3) the variables defined in the closures bound to the ancestor predicates of $X$. For (2) and (3), \systemName guarantees that if either the result-sets of the ancestor predicates or the variables defined in their closures change, operations of $C(X)$ are undone.

\begin{contexample}
The annotated program of Figure \ref{fig:example1} is illustrated in Figure \ref{fig:example1_translated}. It consists of the predicates and their dependent closures. In Figure \ref{fig:example1_translated}, there are three predicates, $P_1$, $P_2$ and $P_3$, where $P_2$ and $P_3$ depend on $P_1$, as they are defined inside the closure of predicate $P_1$. Each predicate in this example has a closure that is represented as a gray box underneath it. The arrow symbol ($=>$) in each predicate represents the execution of the predicate, which returns the result of the evaluation, i.e., a copy of the latest visible version of the requested data object. Then, the returned version is stored in a variable, which is used inside the closure that is bound to the predicate. Inside the closure, the fields of the returned version are accessed and modified. In the case of a modification, a call to the {\em persist} function is necessary in order to store the changes into a new version inside the database, and adding it to the undo buffer of the transaction, as well as the list of versions of the predicate.
\end{contexample}

Predicates are first class citizens  in \systemName. Various types of predicates can be defined, all of which implement the predicate interface comprising of an {\em execute} function and a {\em match} operation.
The execute function accepts a closure as its parameter and binds it to the predicate. This function evaluates the predicate by collecting the visible versions that satisfy the data selection criterion into the result-set of the predicate. Then, the closure gets executed with the result-set and the set of immutable context variables as its parameters. In addition,  the match operation of a predicate checks whether a given version satisfies the selection criterion. The match operation is used during the validation phase of a transaction. Details regarding validation are described in section \ref{sec:validation}. The lifecycle of an \systemName transaction is shown in Figure \ref{fig:predicate_lifecycle}. The details regarding the execution, validation and repair phases for a failed transaction are stated later in this section.

\begin{figure}[htb]
\begin{center}
\leavevmode
\includegraphics[width=0.9\columnwidth]{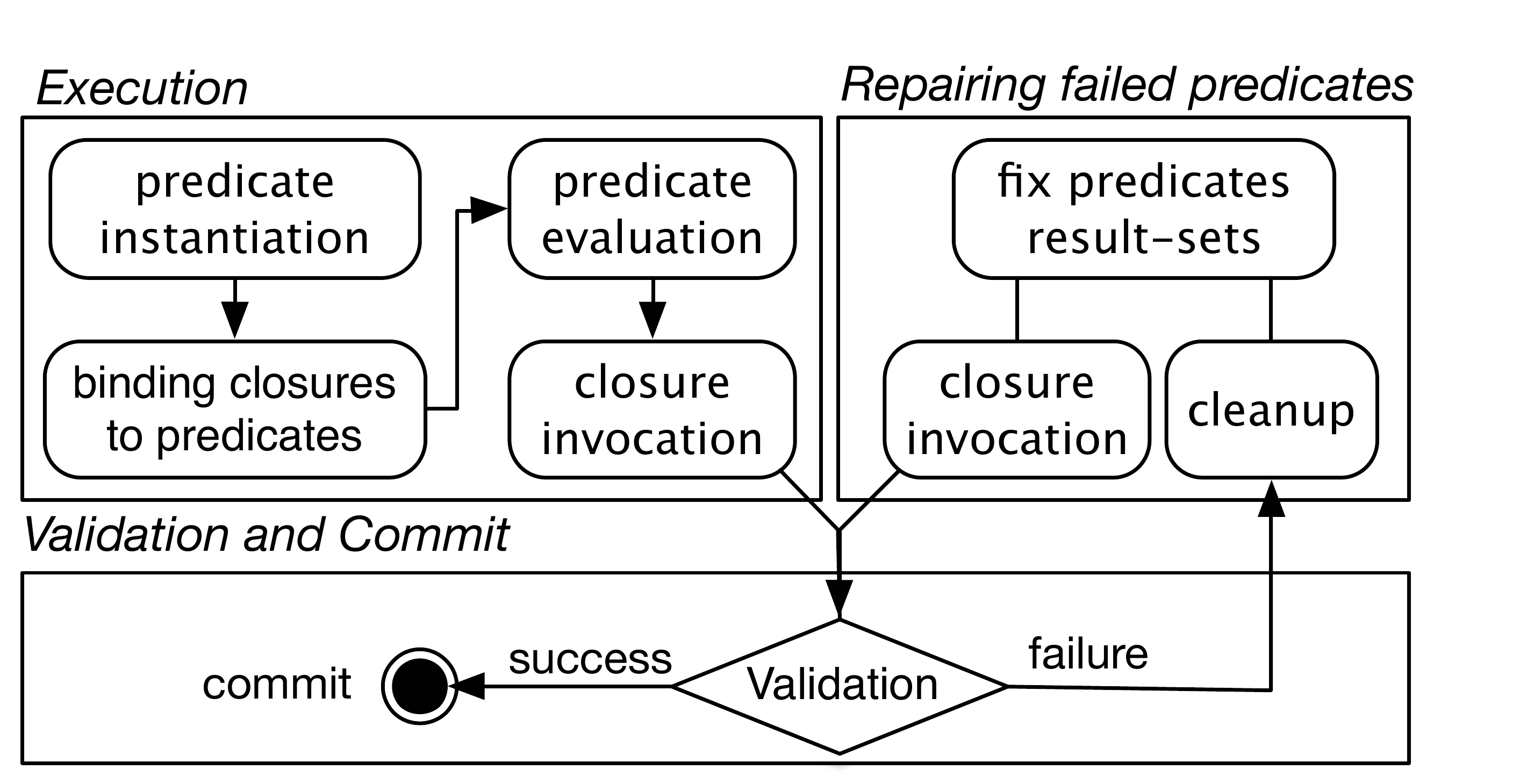}
\end{center}
\vspace{-0.2cm}
\caption{The life cycle of a transaction running under \systemName.}
\label{fig:predicate_lifecycle}
\end{figure}

\designsubsection{\systemName execution}

\begin{figure*}[htb]
\begin{center}
\leavevmode
\includegraphics[width=0.9\textwidth]{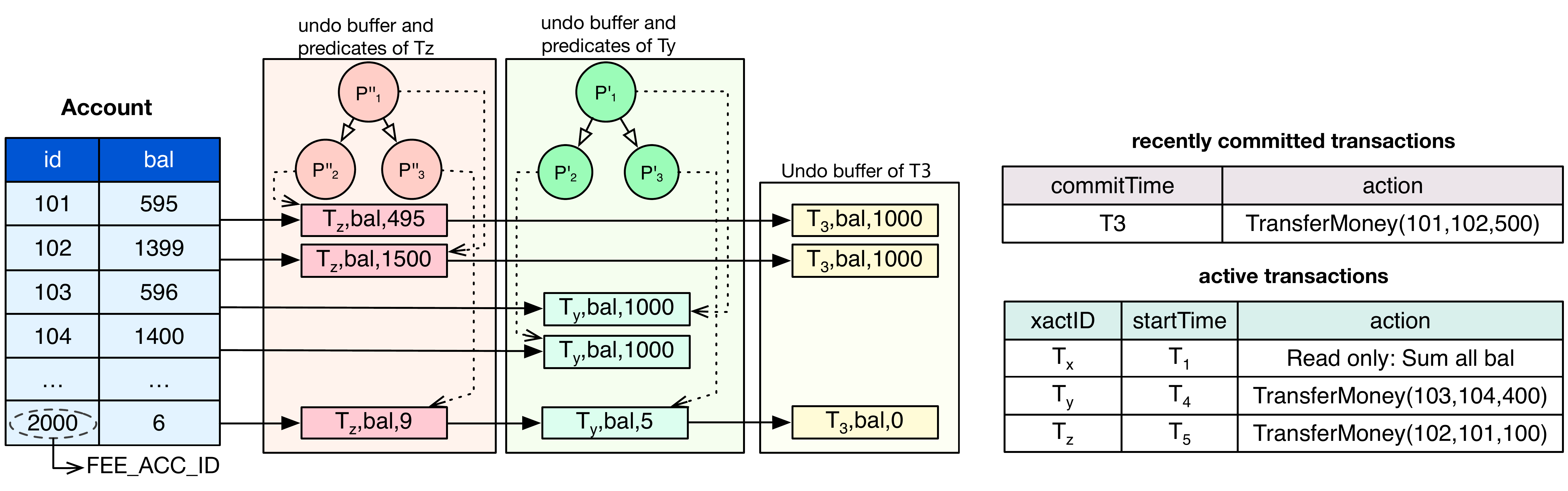}
\end{center}
\vspace*{-0.2cm}%
\caption{A snapshot of the \systemName database for the banking example.}
\label{fig:example1_snapshot}
\end{figure*}

An \systemName program starts by instantiating the root predicates of the predicate graph and calling their execute functions. Each root predicate is evaluated and the closure bound to it is executed. This instantiates the predicates defined in the closure, which get evaluated, and in turn, their closures get executed. 

\begin{contexample}
Figure \ref{fig:example1_snapshot} illustrates a snapshot of the banking example run under \systemName.
In this figure, $P''_1$, $P''_2$ and $P''_3$ are the predicates used in $T_z$. These are the respective runtime instances of predicates $P_1$, $P_2$ and $P_3$ shown in Figure \ref{fig:example1_translated}. The dotted arrows represent the references to their corresponding created versions. The data storage optimization described in section \ref{sec:overview} is used in the version chains.
\end{contexample}

A transaction finishes its execution, when the calls to the execute functions of all the root predicates return successfully. However, some transactions may abort prematurely due to two main reasons. The first reason is receiving a transaction abort command as part of the program after detecting an unwanted state. In this case, repairing or restarting the transaction is irrelevant. The rollback is performed and all the versions created  until that point are discarded. The second reason for a premature abort is having a write-write conflict, described next.

\subsubsection{Handling write-write conflicts} \label{sec:writewrite}

A write-write conflict happens when a transaction tries to write into a data object that has another uncommitted version or a committed version that is newer than the start timestamp of the current transaction. This type of conflict results in a guaranteed failure during the validation phase of \MVCC algorithms such as \NMVCC. As these algorithms do not have a method to recover from such conflicts, they detect them during the execution, abort the transaction prematurely and restart it. However, we strongly believe that if there exists a method for conflict resolution of transactions (such as \systemName), the decision regarding this case should be made either by the transaction programmer or an automated program analyzer.

\begin{example}
Consider a long running transaction which runs successfully until the last statement. The last statement is a write command that writes the results into a highly contended data object. Then, rolling back the transaction just because of a conflict in this last statement might not be the best decision. On the other hand, if a write-write conflict is encountered in the first few statements of the program, and the rest of the program depends on the written data, then continuing the execution from that point would be a complete waste of resources. Here, the repair action to recover from this conflict is similar to restarting it from scratch. The conflicts that happen in predicate $P3$ shown in Figure \ref{fig:example1_translated} are of the former type, while those in predicate $P1$ are of the latter type. 
\end{example}

For this reason, it is logical to provide a configuration option for indicating the behavior in the presence of write-write conflicts. The first option is to deal with a write-write conflict upon its detection by aborting and restarting the transaction early. This is the approach followed by \cite{neumann15}. If the programmer chooses this option, then the transaction is rolled back and restarted from scratch. In this case, reasoning about the state of the closures in order to reuse them is not sound. The second option is to ignore the detection of the write-write conflict and create a new version for the data object. The configuration of the write-write conflict behavior can be done as a table-wide or system-wide setting, which can be overridden for each individual update operation.

Under \systemName, if multiple transactions are allowed to write into a single data object, each creates its own version and adds it to the version chain. Before committing any transaction that is involved in a write-write conflict, these additional versions are not important, as they are  visible only to their own transactions. The approach taken in the validation and commit phase to handle this case without affecting serializability is described in section \ref{sec:wwvalidation}.

\addtocounter{theexample}{-1}
\begin{contexample}
In the snapshot shown in Figure \ref{fig:example1_snapshot}, both $T_z$ and $T_y$ have concurrently written new versions for \checkOK{FEE\_AC-C\_ID}, because write-write conflicts were allowed in TransferMoney. Otherwise, if $T_z$ executes the closure bound to $P''_3$ after $T_y$ finishes all of its operations, $T_z$ prematurely aborts and restarts, because of the uncommitted version written by $T_y$.
\end{contexample}

\designsubsection{\systemName validation and commit} \label{sec:validation}
A transaction that is not aborted prematurely has to be validated before it can be committed. The validity of a transaction is dependent on the validity of its predicates. Other transactions that committed during the execution of the transaction being validated are potentially conflicting with it. Hence, it is necessary to  check if any of the potentially conflicting transactions write to a data object that the current transaction reads from. 

\begin{definition} \label{def:validPred}
\textbf{Valid predicate:} A predicate $P$, belonging to transaction $T$ with start timestamp $S$, is defined to be valid at timestamp $S'$ if and only if the result-set of $P$ is the same when $S'$ is used as the start timestamp of $T$ instead.
\end{definition}

\begin{proposition} \label{prp:predvalid}
A predicate $P$, belonging to transaction $T$, is valid at timestamp $C$ if there is no new version committed between $S$ and $C$ that matches $P$, where $S$ is the start timestamp of $T$.
\end{proposition}

\begin{proofsketch}
The predicate $P$ is evaluated on the snapshot of the database at $S$. The newly committed versions between $S$ and $C$ exist in the undo buffer of the transactions committed in the same time period. Based on Corollary \ref{cor:databasechanges}, these are the only changes that happened to the database in the same period. If $P$ does not match  any newly committed version, then there cannot be a change in its result-set. Thus, based on Definition \ref{def:validPred}, $P$ is still valid at $C$.
\end{proofsketch}

Accordingly, the validation phase consists of matching the predicates of the current transaction against the committed versions of potentially conflicting transactions. If none of the predicates match against any version, then the validation is successful, and the transaction commits. Otherwise, the matched predicates, along with all their descendant predicates, are marked as invalid. For this reason, the predicate graph is traversed in topological sort order, first traversing the higher-level predicates. The \Validation algorithm (Algorithm \ref{alg:validate}) is used by \systemName to validate transactions.

\begin{algorithm}
  \begin{algorithmic}[1]
    \State \textbf{Input}: $G(S)\gets$ predicate graph resulting from executing a 
    
    \hspace{38pt}transaction program $T$ using start timestamp $S$
    \Statex \hspace{35pt}$S'\gets$ validation timestamp, where $S < S'$
    \State \textbf{Output}: $L_1\gets$ list of valid nodes in $G(S)$ using $S'$
    \Statex \hspace{33pt}$L_2\gets$ list of invalid nodes in $G(S)$ using $S'$
    
    \hspace{38pt}along with their descendants
    \State $Q\gets$ the result of applying topological sort on $G(S)$
    \While {$Q$ is non-empty}
        \State $n\gets$remove the node from the head of $Q$
        \If{((a parent predicate of $n$ exists in $L_2$) or 
        \Statex \hspace{23pt} ($n$ is an invalid predicate at $S'$))} 
          add $n$ to tail of $L_2$  \label{alg:validate:l2}
        \Else 
          \hspace{0.5pt} add $n$ to tail of $L_1$ \label{alg:validate:l1}
        \EndIf
    \EndWhile
    \State \Return ($L_1$, $L_2$)
  \end{algorithmic}
\caption{\Validation}
\label{alg:validate}
\end{algorithm}

The \Validation algorithm traverses the predicate graph in order to identify all invalid predicates and their descendants. This is different from \NMVCC where the validation process stops after finding the first invalid predicate, as it cannot succeed anymore. Whereas, in \systemName, the validation phase is additionally responsible for identifying all the invalid predicates. For this reason, the validation process is resumed at the next predicate even though it is known that the transaction cannot succeed validation.
This algorithm not only identifies the invalid predicates, but also topologically sorts and splits the predicate graph into two parts, $L_1$ and $L_2$. All valid predicates are in $L_1$, whereas the invalid predicates and their descendants are in $L_2$. The validation of a transaction, starting at timestamp $S$ and resulting in a predicate graph $G(S)$, is successful at timestamp $S'$, if $L_2$ in the result of \Validation($G(S)$,$S'$) is empty, in which case the transaction commits with commit timestamp $S'$.

\begin{contexample}
In the snapshot  shown in Figure \ref{fig:example1_snapshot}, assume that $T_y$ finishes execution first. Then, as there is no concurrent transaction that committed before $T_y$, it passes validation, its commit timestamp is assigned, and it is added to the recently committed list as $T_6$. Thereafter, $T_z$ enters the validation phase and checks the recently committed list and finds $T_6$ as a concurrent transaction that committed after it had started its execution. Using the \Validation algorithm, $T_z$ traverses the predicate graph in topological sort order and matches each predicate against the committed versions of $T_6$. The \Validation algorithm finds a match for $P''_3$ against $(T_y,bal,5)$, making the predicate invalid. $P''_3$ is the only predicate in this example that gets added to $L_2$, the list of invalid predicates.
\end{contexample}

\begin{lemma}
$L_1$ in the result of the algorithm \Validation($G(S)$,$S'$) (Algorithm \ref{alg:validate}) does not contain any invalid predicate nodes or the descendants of an invalid predicate node, for any given timestamps $S$ and $S'$, and predicate graph $G(S)$, where $S < S'$.
\label{lem:validationL1}
\end{lemma}

\begin{proofsketch}
This follows immediately from lines \ref{alg:validate:l2} and \ref{alg:validate:l1} of the \Validation algorithm.
\end{proofsketch}

\begin{lemma}
$L_1 \concat L_2$, the concatenation of $L_1$ followed by $L_2$, is a topological sort of $G(S)$, where $L_1$ and $L_2$ are the results of the algorithm \Validation($G(S)$,$S'$) (Algorithm \ref{alg:validate}).
\label{lem:validationLinear}
\end{lemma}

\begin{proofsketch}
Algorithm \ref{alg:validate} starts by topologically sorting the predicate graph. Then, the result of the topological sort is divided into two lists, $L_1$ and $L_2$. Since any subset of a topological sort is also topologically sorted, both $L_1$ and $L_2$ individually follow the topological sort. Then, in order to prove that $L_1\concat L_2$ is topologically sorted, it is sufficient to show that there is no node in $L_2$ that has a descendant node in $L_1$. In Algorithm \ref{alg:validate}, the only place where nodes are added to $L_2$ is line \ref{alg:validate:l2}. Based on the condition in line \ref{alg:validate:l2} and the fact that nodes are traversed in topological sort order, if a node is in $L_2$, all of its descendant nodes are also in $L_2$. \end{proofsketch}

\subsubsection{Impact of write-write conflicts on validation} \label{sec:wwvalidation}

Writes by concurrent transactions to the same data object lead to conflicts only in specific scenarios. One such case is when there is a predicate that reads the current value of the data object before updating it. In this case, after one of these transactions commits, the rest fail validation while matching against the committed version. Otherwise, the write is blind and the transaction updates the data object without reading its existing value. However, all writes done by a  transaction, including the blind ones, are visible to the other transactions only after its commit. Therefore, accepting blind writes does not affect serializability, as the read-set of the transaction remains unmodified at commit time, and the illusion of running the entire transaction at commit time is maintained.

Moreover, for keeping the procedure for constructing the visible value simple, \systemName moves a committed version next to the other committed versions in the version chain. This technique preserves the semantics that uncommitted versions are in the beginning of the version chain, and committed ones are ordered by their commit timestamps at the end of the version chain.

\designsubsection{\systemName repair}
If a transaction fails validation, it enters the repair phase. All the read operations in a transaction that runs under a timestamp ordering algorithm (such as \systemName or \NMVCC) return the same result-sets when re-executed, if the start timestamp does not change. In other words, if a transaction reads obsolete data, even if it is re-executed, it would read the same data, and fail validation again. Therefore, the first step for repairing the transaction is picking a new start timestamp $S'$ for the transaction. Then, the \Repair algorithm, shown in Algorithm \ref{alg:repair}, is applied with the parameters ($G(S)$, $S'$), where $G(S)$ is the predicate graph resulting from the initial round of execution using start timestamp $S$. The \Repair algorithm uses the results of  the \Validation algorithm to prune the invalid predicates in the predicate graph.

\begin{algorithm}
  \begin{algorithmic}[1]
    \State \textbf{Input}: $G(S)\gets$ predicate graph resulting from executing a 
    
    \hspace{38pt}transaction program $T$ using start timestamp $S$
    \Statex \hspace{35pt}$S'\gets$ validation timestamp, where $S < S'$
    \State \textbf{Output}: $G'(S')\gets$ repaired predicate graph
    \State $(L_1,L_2)\gets$ \Validation($G(S)$,$S'$) \label{alg:repair:l0}
    \State $F\gets$ Set of all nodes in $L_2$ with no incoming edges from $L_2$ \label{alg:repair:l1}
    \ForAll{predicate node $f$ in $F$} \label{alg:repair:l5}
       \ForAll{predicate node $h$ in descendant nodes of $f$} \label{alg:repair:l2}
         \State clear the list of versions in $h$ and remove them from the
         \Statex \hspace{35pt} undo buffer of the transaction \label{alg:repair:l7}
         \State remove $h$ from $G(S)$ \label{alg:repair:l3}
       \EndFor
       \State clear and remove list of versions in $f$  and remove them
       \Statex \hspace{23pt} from the undo buffer of the transaction\label{alg:repair:l8}
       \State call $f.execute(C(f), S')${\small , where $C(f)$ is the closure bound to $f$}\label{alg:repair:l4}
    \EndFor \label{alg:repair:l6}
    \State \Return $G(S)$ as $G'(S')$
  \end{algorithmic}
\caption{\Repair}
\label{alg:repair}
\end{algorithm}

In line \ref{alg:repair:l1} of the \Repair algorithm, the set of nodes in $L_2$ with no incoming edges from $L_2$ are selected into $F$, where $L_2$ is the second part of the output of the \Validation algorithm. Based on the selection criterion, it is guaranteed that the nodes in $F$ are not descendants of each other. It can be observed in lines \ref{alg:repair:l2} to \ref{alg:repair:l3} that a given invalid predicate $f$ in $F$ is {\em pruned}, which is the process of removing all the versions created by $f$ and its descendants, and then removing the descendants from the predicate graph. After pruning the predicate node $f$, it is sufficient to  re-execute $C(f)$, which is done in line \ref{alg:repair:l4}. Executing the closure re-instantiates all  pruned descendant predicates. The order of re-executing the invalid predicates does not matter, as they are independent. If one is a descendant of the other, the former would have been removed from the graph during the pruning process. In the \Repair algorithm, the closures of the valid predicates are not re-executed, as there are no changes in their result-sets.

\begin{contexample}
After $T_z$ fails validation, a new start timestamp $T_7$ is assigned to it. Based on the validation results, $P''_3$ is the only invalid predicate in $T_z$. $P''_3$ reads the version written by $T_6$ with $bal=6$, puts in-place value $bal=10$ in the table and creates $(T_z,bal,6)$ as its version.
\end{contexample}

After applying the \Repair algorithm, the transaction enters the validation phase again. The validation phase is the same as the one for the initial execution of the transaction and has a possibility of success or failure.

\begin{contexample}
In the above example, $T_z$ succeeds in the validation phase this time, as there is no concurrent transaction committed during its new lifetime (after $T_7$).
\end{contexample}

\begin{lemma} \label{lem:prune}
Given a predicate graph  $G(S)$ and an arbitrary node $X$ in it, pruning $X$ and re-executing $C(X)$, the closure bound to $X$, under the same start timestamp $S$, rebuilds $G(S)$.
\end{lemma}

\begin{proof}
Observe that both execution and re-execution of $C(X)$ view the same snapshot of the database, as the start timestamp of the transaction is not changed. Since all the versions created by $X$ and its descendants are pruned, any changes created by the closures bound to $X$ and its descendants are removed. In addition, C(X), based on Definition \ref{def:closure}, is a deterministic function, which guarantees given the same input it generates the same output. Therefore, re-executing $C(X)$ re-creates the same predicate graph.
\end{proof}

\begin{lemma} \label{lem:repair}
Let $G(S)$ be the predicate graph resulting from the execution of a transaction $T$ with start timestamp $S$ that failed validation at timestamp $S'$. Assume that $G'(S')$ is the predicate graph resulting from applying the \Repair algorithm on ($G(S)$, $S'$). Instead of applying the \Repair algorithm, if $T$ was aborted and restarted with start timestamp $S'$, resulting in a predicate graph $G''(S')$, then $G''(S')$ is equivalent to $G'(S')$.
\end{lemma}

\begin{proof}
Assume that the database is duplicated at timestamp $S'$ into two identical instances, $D_1$ and $D_2$. On instance $D_1$, $T$ aborts and restarts using the new start timestamp $S'$, resulting in the predicate graph $G''(S')$. On instance $D_2$, the \Repair algorithm is applied on ($G(S)$, $S'$), which results in the predicate graph $G'(S')$. It should be noted that both $G'(S')$ and $G''(S')$ have the same root predicates, without considering their lists of versions. The reason is that the creation of a predicate depends on its parent predicate. The root predicates do not have a parent by definition, so they are created regardless of the database state. It is also guaranteed that they are also not removed from the predicate graph by the \Repair algorithm, as they are not descendants of any other nodes. Moreover, in line \ref{alg:repair:l0} of the \Repair algorithm, the results of the \Validation algorithm on ($G(S)$, $S'$) are used, i.e., $L_1$ and $L_2$. Let $X$ be an arbitrary node selected from $G(S)$. There are three cases.

\textbf{Case 1:} $X$ is in $L_1$. Then, $X$ is a valid predicate, and based on Lemma \ref{lem:validationL1}, all of its ancestors are valid, too. Let $Y$ be the list of root predicates that are the ancestors of $X$. Nodes in $Y$ are common in both $G'(S')$ and $G''(S')$. Hence, as part of executing $T$ on $D_1$, the corresponding nodes in $G''(S')$ are created and the closures bound to them are executed. As nodes in $Y$ are valid nodes, they read the same data as $S$, and take the same program flow, which results in creating $X$ with the same list of versions inside it.

\textbf{Case 2:} $X$ is in $L_2$ and has no incoming edge from another node in $L_2$. Then, either $X$ has no parent,  or its parents are in $L_1$. In both cases, $X$ is created while executing $T$ from scratch on $D_1$. If $X$ has no parent, then it exists regardless of executing under timestamp $S$ or $S'$. If $X$ has a list of parents $Y$ in $L_1$, then based on Case 1, all nodes $Y$ exist in $G''(S')$ and they are valid. As the execution of their closures using the timestamp $S'$ is the same as $S$, they create the same child predicates and $X$ is among them. Lines \ref{alg:repair:l5}-\ref{alg:repair:l6} of the \Repair algorithm prune $X$ and re-execute it. Then, the closure bound to $X$ is re-executed in line \ref{alg:repair:l4} of the \Repair algorithm using timestamp $S'$, which is identical to executing it on $D_1$. And, it results in the same descendant nodes and list of versions for $X$ in both $G'(S')$ and $G''(S')$.

\textbf{Case 3:} $X$ is in $L_2$ and has an incoming edge from another node in $L_2$. In this case, $X$ has an ancestor, $Z$, from $L_2$ that falls into Case 2. Thus, $Z$ is pruned and $X$ is removed from $G(S)$ in lines \ref{alg:repair:l7} and \ref{alg:repair:l8} of the \Repair algorithm. As it is shown in Case 2, the same descendant nodes are recreated for $Z$ in both $G'(S')$ and $G''(S')$.

Thus, $G'(S')$ and $G''(S')$ are equivalent.
\end{proof}

In the case of a validation failure in this stage, another round of repair is initiated and this cycle continues until validation is successful. It is important to note that the whole process of validating a transaction, and drawing a commit timestamp or a new start timestamp depending on the result of the validation is done in a short critical section. However, the repair phase is done completely outside the critical section, as if the transaction is running concurrently with other transactions.

%% file: theory.tex
\designsubsection{Serializability Proof}

We prove that the schedules created using \systemName are commit-order serializable. Consequently, it is  guaranteed not only that the schedules created by \systemName are serializable, but also that one equivalent serial schedule can be proposed by creating a serial sequence of transactions using their commit order under \systemName.

\begin{theorem}
The committed projection of any multi-version schedule H produced under \systemName is conflict equivalent to a serial single-version schedule H'  where the order of transactions in H' is the same as the order of successful commit operations in H and uncommitted transactions are ignored.
\end{theorem}

\begin{proof}
Based on Definition \ref{def:vdv} and Proposition \ref{prp:undobuffer}, the versions created by a transaction become visible to the other transactions only after the transaction commits. Accordingly, the operations done by uncommitted transactions cannot affect the serializability of the committed transactions. For this reason, the uncommitted transactions can safely be ignored, and only the committed transactions are considered in this proof.

By showing that all the dependencies among different transactions executed under \systemName have the same order as their commit timestamp, it can be proven that any execution of transactions under \systemName is serializable in commit order. Read-only transactions read all committed versions that are committed prior to their start. Moreover, the commit timestamp of a read-only transaction is the same as its start timestamp, as if it is executed at that  point in time.

An update transaction starts, gets executed, and then enters the validation and commit phase. If the validation is successful, a commit timestamp $C$ is assigned to the transaction. Otherwise, the transaction acquires a new start timestamp $S'$ and enters the repair phase. As it is already shown in Lemma \ref{lem:repair}, the predicate graph resulting from applying the \Repair algorithm at $S'$ is the same as the one resulting from aborting and restarting the transaction from scratch at timestamp $S'$. Then, the repaired transaction enters the validation phase again, and this cycle continues until the transaction succeeds in the validation phase.

Thus, it is sufficient to prove that for an update transaction that starts with timestamp $S$, executes, successfully passes the validation phase and commits, the visible effect of the transaction is as if it is executed at one point in time, $C$, where $S < C$. The proof is done by contradiction, and is similar to the serializability proof in \cite{neumann15}. There are some modifications for allowing write-write conflicts in \systemName, as this kind of conflicts leads to premature abort and restart in \cite{neumann15}.

Let $T_1$ be an update transaction from the committed projection of $H$ with start timestamp $S_1$ and commit timestamp $C_1$. Suppose that the execution of $T_1$ cannot be delayed until $C_1$. Then, there is an operation, $o1$ in $T_1$ that conflicts with an operation $o2$ in another transaction, $T_2$ (started at $S_2$ and committed at $C_2$), with $o1 < o2$ and $T_2$ committed during the lifetime of $T_1$, i.e., $C_2 < C_1$. There are four cases corresponding to the possible combinations of two operations $o1$ and $o2$.

\textbf{Case 1:} both are reads. In this case, $o1$ and $o2$ can be swapped and this contradicts the assumption that $o1$ and $o2$ are conflicting.

\textbf{Case 2:} $o1$ is write and $o2$ is read. Based on Definition \ref{def:vdv}, the version written by $o1$ is not visible to $o2$, as $S_2 < C_1$. Thus, it contradicts the assumption that $o1$ and $o2$ are conflicting.

\textbf{Case 3:} $o1$ is read and $o2$ is write. By $C_1$, the version $V_2$ written by $o2$ is committed and exists in the undo buffer of $T_2$, as $C_2 < C_1$. The operation $o1$ is done via a predicate $P_1$. Hence, while validating $T_1$ at $C_1$, the \Validation algorithm matches $P_1$ against $V_2$. If $P_1$ matches, $T_1$ fails validation and cannot commit at $C_1$. This contradicts the fact that $T_1$ with this commit timestamp is in the committed projection of $H$. Otherwise, if $P_1$ does not match $V_2$, then it contradicts the assumption that $o1$ and $o2$ are conflicting.

\textbf{Case 4:} both are writes. In this case, $o1$ and $o2$ can be swapped and this contradicts the assumption that $o1$ and $o2$ are conflicting. Basically, an individual write operation acts as a blind-write. Thus, the blind-write operation in $o1$ can be delayed until $C_1$. If a write operation is not a blind-write, there is a read operation done before it, and this case is converted to either Case 2 or Case 3.
\end{proof}

%% file: interop.tex
\section{Interoperability with MVCC} \label{sec:interop}

The only interaction of an \systemName transaction with other transactions is during the validation phase. The algorithm needs to know about the  versions committed during the lifetime of the current transactions, so as to match the predicates against them. This is the only information that is needed from a previously executed transaction.
Consequently, any MVCC algorithm, such as \NMVCC, that provides information about the committed versions can seamlessly inter-operate with \systemName. This interoperability provides backward compatibility for free, and makes it possible for the transaction developers to program their new transactions in \systemName, and gradually convert the old transactions in the system into \systemName.

%% file: optimizations.tex
\section{Optimizations} \label{sec:optimizations}

In addition to the generic design of \systemName, several optimizations are possible in order to improve transaction execution. The rest of this section is dedicated to optimizations that can be employed by this concurrency control algorithm.

\subsection{Attribute-Level Predicate Validation}

In order to validate an \systemName transaction, the committed versions of all concurrent transactions are considered. These versions are matched (as the whole \highlight{record}) against each predicate in the predicate graph. If a version matches with a predicate, a conflict is \highlight{declared}. However, it is a pessimistic approach, as it is possible that the modified columns in the concurrently committed versions are not used in the current transaction.

As an optimization, the validation can be done at the attribute-level instead. This optimization has already been used in previous works \cite{neumann15}. In order to enable the attribute-level validation, the columns that are used from the result-sets of the predicates are marked for runtime monitoring. In addition, all columns in the data selection criterion of the predicate are monitored. At runtime, each created version stores a list of columns modified in it. Then, in the validation phase, while checking a predicate against a version, the intersection between the monitored columns in the predicate and the modified columns in the version is first computed. If the intersection is empty, there cannot be a match. Otherwise, the match operation that is predicate-specific is performed. 

\subsection{Reusing Previously Read Versions}\label{sec:reuse}
The predicates that failed validation are the starting points of the repair phase. After acquiring a new start timestamp in this phase, the failed predicates are re-evaluated, and their result-sets are fed into their assigned closures. Evaluating these predicates from scratch is one approach and it is taken by default in \systemName. However, it is also possible to re-use some of the computation done in the initial execution round that failed validation. 

Even though the realization of this optimization depends on the type of predicates and the way they access their target data, there are some similarities among all of them. Each predicate that wants to re-use the computation keeps a reference to its result-set. The result-set has to be computed nonetheless, as it needs to be fed into the closure of the predicate. This additional reference is kept only until a successful commit, which is also the end of the lifetime of the predicate. Then, in case of a failed validation, the result-sets have to be \highlight{fixed} for the failed predicates. The procedure to fix the result-set is predicate-specific. This procedure is merged with the validation phase, and is done at the same time, as both require matching predicates against versions. 

\addtocounter{theexample}{1}
\begin{example}
Consider a predicate that selects data based on a condition over non-indexed columns of a large table. In this case, the initial execution is costly, as a full scan over the whole table is required. However, if this predicate fails, the only reason for this failure is another concurrent transaction committing a version that should be a part of the result-set of the predicate. Therefore, the result-set can be fixed by accommodating the concurrently committed versions into the result-set. 
\end{example}

This optimization comes with a cost for keeping a reference to the result-set of each predicate, as well as the procedure for fixing them. Therefore, this optimization is not used by default, and can be enabled per predicate instance. The decision regarding this optimization is taken either by the transaction programmer, or an automated analyzer that monitors the failure rate of each predicate. One heuristic approach is to activate this optimization only for the predicates that have a higher failure probability, where fixing the result-set is cheaper than re-evaluation.







%% file: impl.tex
\section{Implementation} \label{sec:impl}

\begin{figure*}[htb]
        \begin{subfigure}[t]{0.48\textwidth} \centering
                \includegraphics[width=\columnwidth]{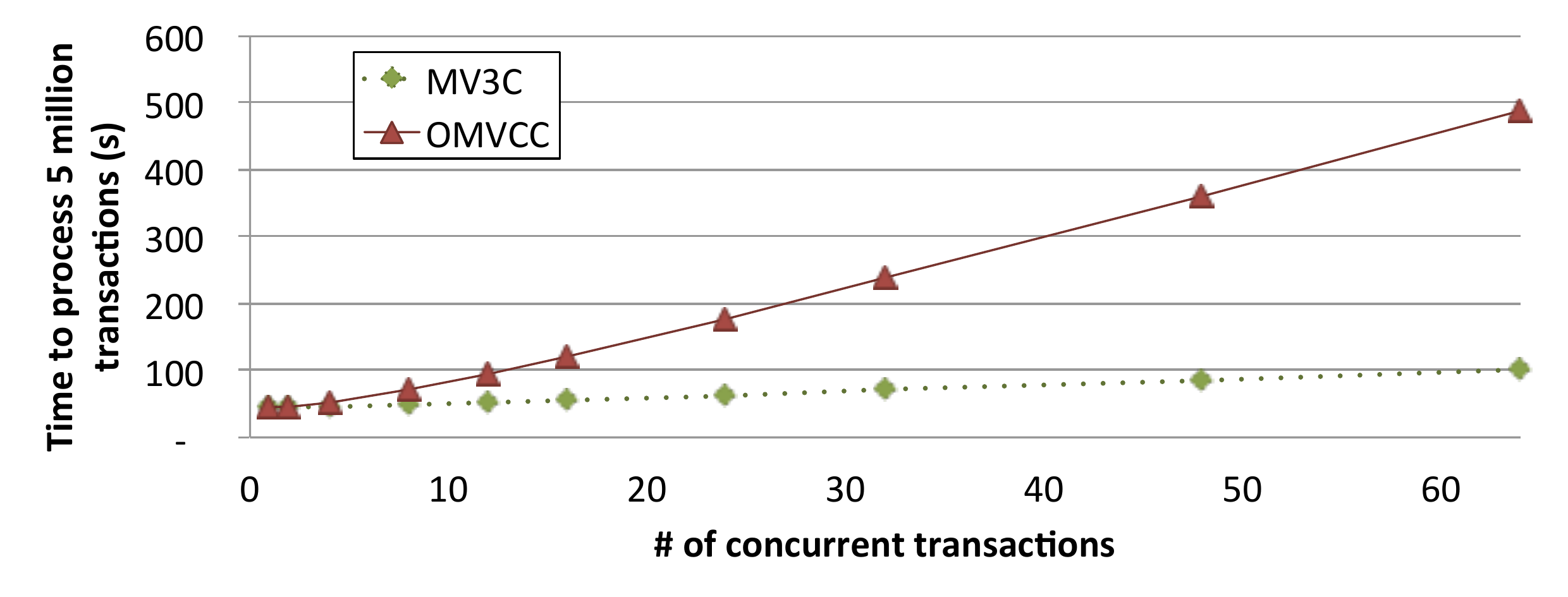}
                \caption{Impact of the number of concurrent transactions}
                \label{fig:trading_window}
        \end{subfigure}%
        \hfill
        \begin{subfigure}[t]{0.48\textwidth} \centering
                		\includegraphics[width=\linewidth]{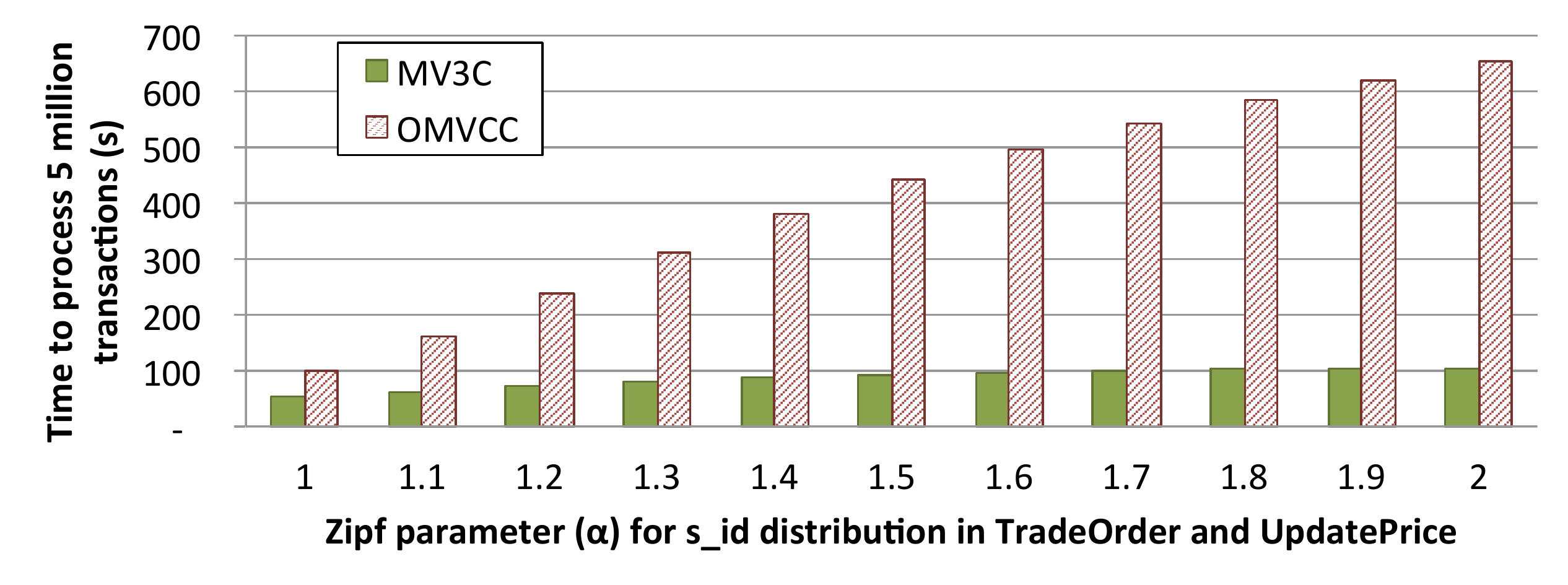}
                \caption{Impact of the percentage of conflicting transactions}
                \label{fig:trading_percentage}
        \end{subfigure}%
        \vspace*{0.6cm}%
         \caption{Trading benchmark experiments}
         \label{fig:trading_exp}
        \vspace*{-0.2cm}%
\end{figure*}

\systemName is more of a concept rather than a specific and strict concurrency control technique, just like \MVCC. Thus, there can be different flavors of \systemName, each with its own optimizations. 
In this section, the details of our implementation of \systemName are described.

\pheader{Transaction Management}
\systemName has a {\em transaction manager} that is responsible for starting and committing the transactions, similar to the one in \cite{neumann15}. It primarily stores four data items which are shared among all transactions, namely {\em recently committed} transactions, {\em active} transactions, {\em start-and-commit timestamp sequence generator}, and {\em \highlight{transaction identifier} sequence generator}.

The active transactions contains the list of ongoing transactions that have not committed. This list is updated whenever a transaction starts or commits. It is mainly used for tracking the active transaction with the oldest start timestamp, which is necessary for garbage collection. Like \cite{neumann15}, the garbage collection of the versions created by a committed transaction is performed after ensuring that there is no older active transaction that could read the versions. The start-and-commit timestamp sequence generator is used for issuing start and commit timestamps. Since both timestamps are issued from the same sequence, in order to find the transactions that ran concurrently with a given transaction $T$, it is sufficient to choose transactions from the recently committed list for which the commit timestamp is greater than the start timestamp of $T$.

In addition, the transaction identifier sequence generator  assigns a unique identifier to each transaction. This unique identifier plays the role of a temporary commit timestamp for an active transaction. The sequence generator starts from a very large number, larger than the commit timestamp of any transaction that can appear in the lifetime of the system. This timestamp is used in the uncommitted versions written by an active transaction, and simplifies the distinction between committed and uncommitted versions.


Starting a transaction using the transaction manager is done by drawing a start timestamp and a transaction identifier from the corresponding sequence generators. Then, a transaction object is created with these values. The transaction object is an encapsulation of the data needed for running, validating and committing a transaction. Besides the start timestamp and the transaction ID, the transaction object keeps track of the undo buffer and the predicate graph. The predicate graph is implemented as a list of root predicates with each predicate storing a list of its child predicates. All operations that interact with the database for reading or writing data require the transaction object, which is passed as a parameter to the operations. All data manipulation operations result in one or more versions, and references to these versions are kept in the undo buffer.

\pheader{Concurrent execution} In this implementation, it is assumed that the transaction programs are divided into smaller pieces and the concurrent execution of different transactions is realized by interleaving these small pieces of different transactions. Then, at each step, a single transaction piece gets executed, as the implementation is single threaded. This approach can easily be extended to a multi-threaded implementation where each operation on the data (shared among different threads, each one running its own transaction) is done in a critical section, where only a single thread can perform data manipulation. A lock-free multi-threaded implementation of \systemName is an interesting future work.

%% file: eval.tex
\section{Evaluation} \label{sec:eval}

We use the TATP \cite{tatp} and TPC-C \cite{tpcc} benchmarks, as well as some of our own mini benchmarks in order to showcase the main pros and cons of \systemName compared to \NMVCC \cite{neumann15}. We have implemented both algorithms in C++11. The size of the implementation, excluding the tests and benchmarks, for \systemName is 4.5 kLOC and that for \NMVCC is 3 kLOC. In both of the implementations,  a row-based storage is used, and redo logs are stored in memory.

\begin{figure*}[htb]
        \begin{subfigure}[t]{0.32\textwidth} \centering
                \includegraphics[width=\linewidth]{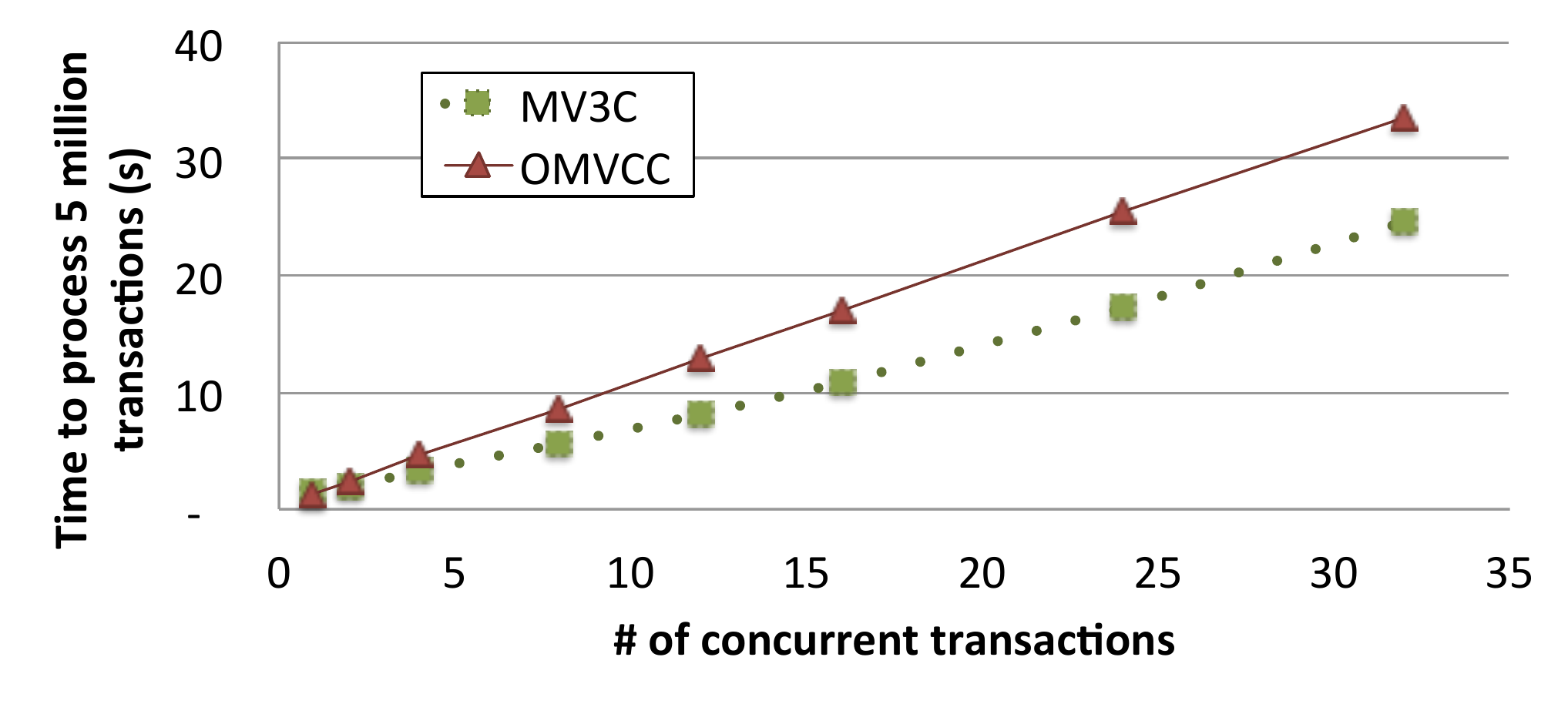}
                \caption{Impact of the number of concurrent transactions with 10 million rows in Account table}
                \label{fig:banking_window}
        \end{subfigure}%
        \hfill
        \begin{subfigure}[t]{0.32\textwidth} \centering
                		\includegraphics[width=\linewidth]{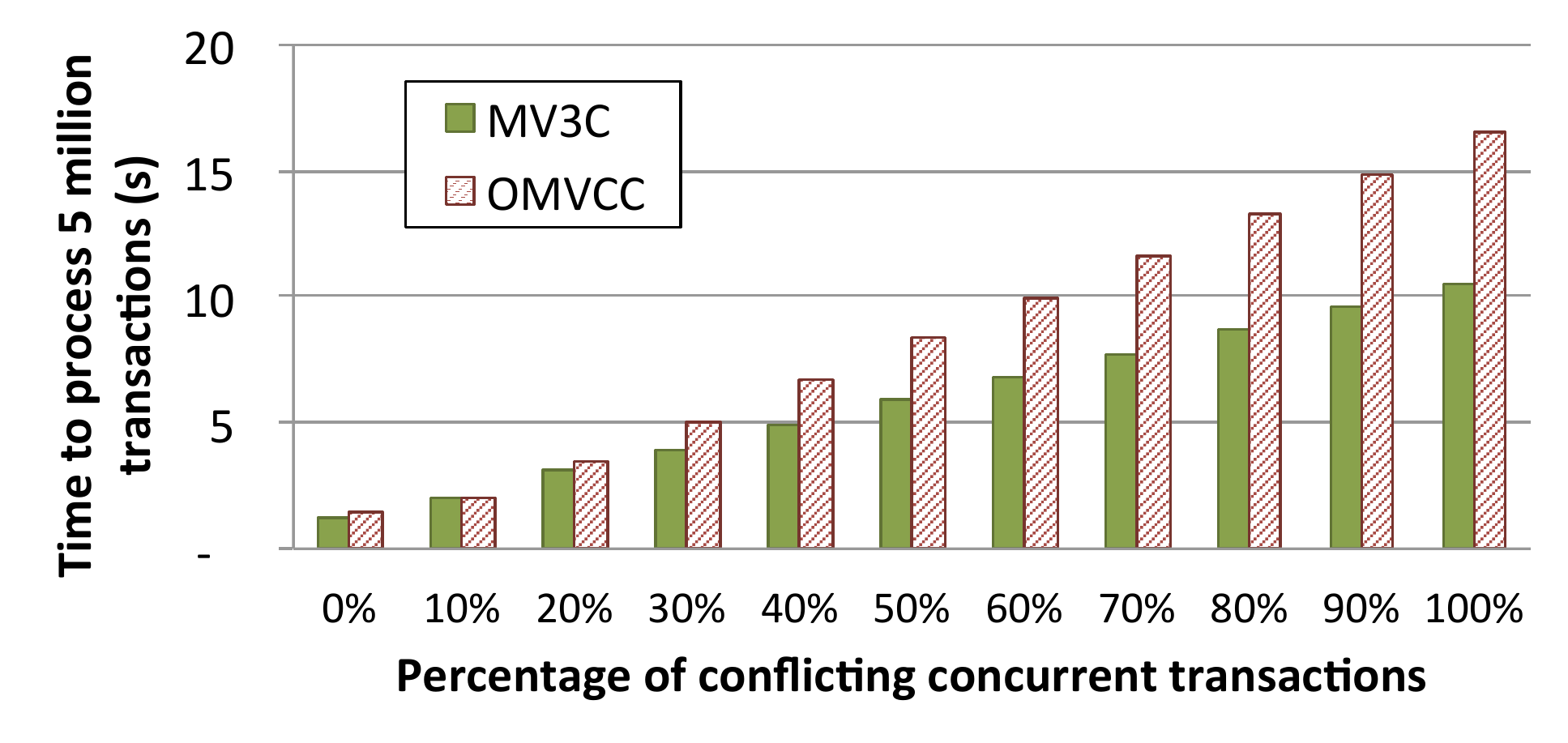}
                \caption{Impact of the percentage of conflicting transactions with 10 million rows in Account table}
                \label{fig:banking_percentage}
        \end{subfigure}%
        \hfill
        \begin{subfigure}[t]{0.32\textwidth} \centering
                \includegraphics[width=\linewidth]{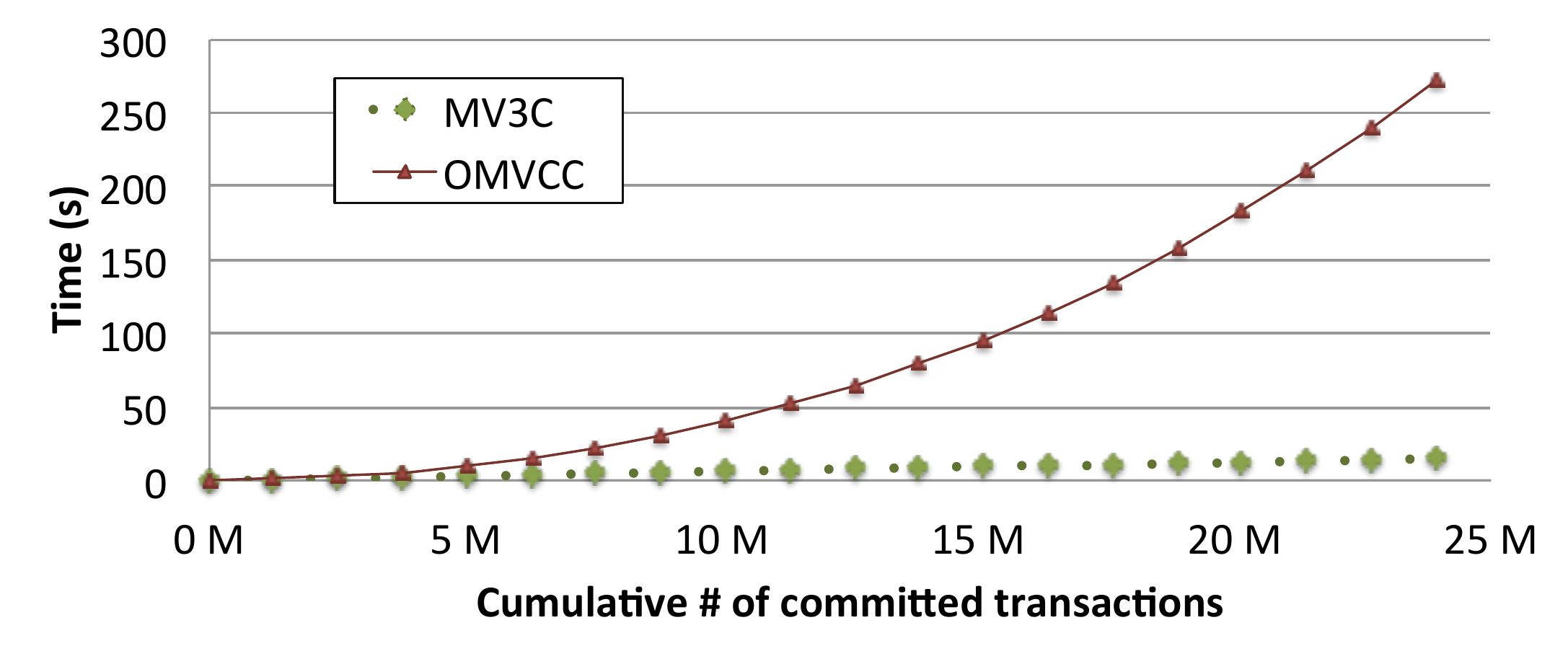}
                \caption{The ripple effect with 100 million rows in Account table}
                \label{fig:banking_ripple}
        \end{subfigure}%
        \vspace*{0.7cm}%
        \caption{Banking example experiments}
        \label{fig:banking_exp}
\end{figure*}

The implementations of both algorithms are single-threaded. The concurrency among transactions is realized in these implementations by dividing programs into smaller pieces and interleaving these pieces with the ones from other programs. The same approach is taken in \cite{neumann15} while evaluating \NMVCC. A program can be divided into at least three pieces: (1) the start transaction command, (2) the logic of the transaction program itself, and (3) the commit transaction command. In addition, the logic of the transaction program can be further divided into smaller pieces. However, the concurrency level only depends on the number of active transactions (i.e., the transactions that started but not committed) rather than the number of interleaved program pieces. In our experiments, unless otherwise stated, we use the notion of {\em window} to control the number of concurrent transactions. Given a window size $N$, which implies having $N$ concurrent transactions, $N$ transactions are picked from the input stream and all of them start, then they execute, and finally they try to validate and commit one after the other. Here, the transactions that fail during the execution are rolled back and moved to the next window. Furthermore, the transactions that fail validation acquire a new timestamp immediately and their executions are moved to the next window. It should be noted that $N=1$ is equivalent to a serial execution of transactions.

All experiments are performed on a dual-socket Intel\textsuperscript{\circledR}  Xeon\textsuperscript{\circledR} CPU E5-2680 2.50GHz (24 physical cores) with 30 MB cache, 256 GB RAM, Ubuntu 14.04.3 LTS, and GCC 4.8.4. Hyper-threading, turbo-boost, and frequency scaling are disabled for more stable results. We run all the benchmarks for 5 times on a simple in-memory database engine and report the average of the last 3 measurements. In all the experiments, the variance of the results was found to be less than 1\% and is omitted from the graphs. 

\subsection{Rollback vs. Repair}

The main advantage of \systemName compared to \NMVCC is repairing a conflicting operation instead of aborting and restarting the transaction. There are different parameters that impact the effectiveness of both \systemName and \NMVCC, such as the number and size of concurrent transactions, the percentage of conflicts among concurrent transactions, and the average number of times that a transaction gets aborted until it commits. In the following subsections, we focus on showing the impact of these parameters on the effectiveness of both \systemName and \NMVCC. For this purpose, the TATP benchmark and  two other mini-benchmarks are used. The first mini-benchmark is  based on our banking example (Example \ref{ex:banking}), and the second one, named {\em Trading}, is described next.

\begin{example} \label{ex:trading}
\pheader{Trading benchmark} This benchmark simulates a simplified trading system, and consists of four tables:
\begin{itemize}
\item Security(\underline{s\_id}, symbol, s\_price): the table of securities available for trading along with their prices. This table has 100,000 records.
\item Customer(\underline{c\_id}, cipher\_key): the table of customers along with their encryption/decryption key, used for secure personal data transfer. This table has 100,000 records.
\item Trade(\underline{t\_id}, t\_encrypted\_data): the table for storing the list of trades done in the system. The trade details are stored in encrypted form using the customer's cipher\_key, and consists of a timestamp of the trade. This table is initially empty.
\item TradeLine(\underline{t\_id, tl\_id}, tl\_encrypted\_data): the table of items ordered in the trades. The details of each record in the last column are encrypted using the customer's cipher\_key and consist of the identifier of an asset, and the traded price. The traded price is negative for a buy order. This table is initially empty.
\end{itemize}

Moreover, there are two transaction programs in this benchmark: TradeOrder and PriceUpdate. A TradeOrder transaction accepts a customer ID (c\_id) and an encrypted payload containing the order details. First, the customer's cipher\_key is read. Next, using the customer's cipher\_key, the payload is decrypted, which contains a sequence number for t\_id, timestamp of the order, and the list of securities along with buy or sell flags for each security. Then, the corresponding securities are read from the Security table so as to get their current prices and then, the respective rows are added to the Trade and the TradeLine tables. A PriceUpdate transaction updates the price of a given security from the Security table.

The instances of these two transaction programs conflict, if a security is requested in a TradeOrder, while concurrently a PriceUpdate is updating its price. In order to simulate the different popularities among securities, the sec\_id input parameters in both transaction programs are generated following Zipf distribution.
\end{example}

\subsubsection{Number of concurrent transactions}
Given a fixed percentage of conflicting transactions in a stream of transactions, if we change the number of concurrent transactions, more conflicts occur. As it was mentioned earlier, we use the notion of {\em window} for controlling the number of concurrent transactions.

\subfigref{trading_exp}{trading_window} shows the results of the Trading benchmark (Example \ref{ex:trading}) with different window sizes. In this experiment, the  distribution parameter $\alpha$ for generating sec\_ids is 1.2 for both TradeOrder and PriceUpdate transactions. As shown in this figure, \systemName processes a fixed number of transactions faster and faster compared to \NMVCC, as the window size increases. \systemName avoids decryption and deserialization of the input data, and re-executes only a small portion of the conflicting TradeOrder transactions, while \NMVCC re-executes the conflicting transactions from scratch. Additionally, PriceUpdate consists of a blind write operation, which does not lead to a conflict in \systemName, but creates a conflict in \NMVCC.

\subfigref{banking_exp}{banking_window} shows the results of a similar experiment for the Banking example (Example \ref{ex:banking}). In this experiment, there are only TransferMoney transactions, all of which conflict at the end while updating the fee account (line \ref{fig:example1:updateFeeLine} in \figref{example1}). The figure shows the effectiveness of \systemName compared to \NMVCC as the time gap for executing 5 million transactions gets wider for bigger windows. However, the behavior is different from the Trading experiment in \subfigref{trading_exp}{trading_window} due to three reasons. First, TransferMoney is a light-weight transaction, and its re-execution in \NMVCC is not significantly more expensive than partial re-execution in \systemName. Second, these conflicts lead to a validation failure in \systemName, while they lead to a premature abort in \NMVCC, without going through the validation. Third, beyond a certain number of concurrent transactions, the performance of \systemName deteriorates due to the effect of increased version chain length. In \NMVCC, the version chain length of the data object identified by $FEE\_ACC\_ID$ is at most one, due to write-write conflicts. Whereas in \systemName, as write-write conflicts are handled, there are as many uncommitted versions as the number of concurrent transactions, all of which have to be traversed until the visible version is found.

\begin{figure}[htb]
\begin{center}
\leavevmode
\includegraphics[width=\columnwidth]{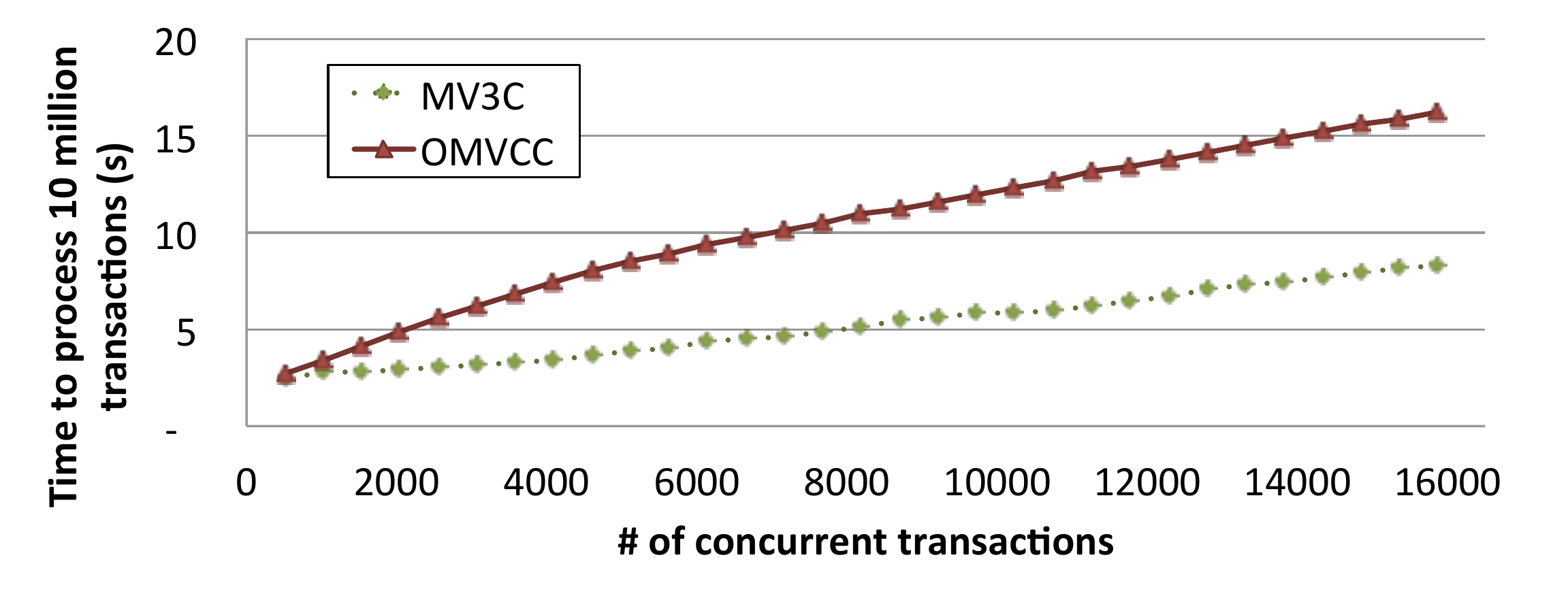}
\end{center}
\vspace*{-0.2cm}
\caption{TATP experiment	with scale factor 1}
\label{fig:tatp}
\end{figure}

In addition, we have implemented the TATP benchmark in both \systemName and \NMVCC. \figref{tatp} shows the results of executing 10 million TATP transactions with different window sizes. This experiment is done for scale factor 1 with non-uniform key distribution, and attribute-level validation is used. The input data for the experiment is generated using OLTPBench \cite{oltpbench}. Once again, this experiment confirms the increased effectiveness of \systemName compared to \NMVCC, as the window size increases. It should be noted that even with non-uniform key distribution, the number of conflicting transactions for small window sizes is low, as 80\% of the workload consists of read-only transactions. Thus, both \systemName and \NMVCC show similar results for small window sizes, and the difference can only be seen in bigger windows. Moreover, as transaction programs in the TATP benchmark are very small, the main advantage of \systemName compared to \NMVCC is that the former optionally accepts write-write conflicts, while the latter prematurely aborts after a write-write conflict. This decision leads to having no conflicts among Update\_Location transaction instances in \systemName, while it leads to aborts in \NMVCC.

\subsubsection{Percentage of conflicts among transactions}

For a fixed number of concurrent transactions, the percentage of conflicting transactions determines the success rate in a window. If this percentage is 0\%, then there are no conflicts. On the other hand, if this percentage is 100\%, only one transaction  succeeds in each window, the rest fail, and are moved to the next window. Consequently, it is expected that \systemName is more effective compared to \NMVCC when the percentage of conflicting transactions is higher.

One experiment to illustrate the impact of the percentage of conflicting transactions uses the Banking example. Another program, called NoFeeTransferMoney, is introduced which is similar to TransferMoney, but without the fee. In this experiment, the window size is fixed to be 16, and the percentage of NoFeeTransferMoney programs is varied keeping the total number of programs fixed. \subfigref{banking_exp}{banking_percentage} shows that \systemName is more effective compared to \NMVCC as the percentage of conflicting transactions increases.

As another experiment, different $\alpha$ parameters of the Zipf distribution of s\_id input parameters in Trading benchmark (Example \ref{ex:trading}) are used with the fixed windows size of 32. The $\alpha$ parameter in this experiment determines the percentage of conflicting transactions. The results of this experiment are illustrated in \subfigref{trading_exp}{trading_percentage}. Once again, this figure shows the desired behavior for \systemName, as the processing time of 5 million transactions becomes significantly lower compared to \NMVCC, as the percentage of conflicting transactions is increased by a larger $\alpha$ parameter.

\subsubsection{Ripple Effect}

The time saved  by repairing the transactions instead of aborting and restarting them from scratch can be used to reduce the load on the system. To show this effect, we have designed an experiment where there are two streams issuing transactions at a constant rate. The first stream issues its transactions at almost the transaction processing rate, while the second stream issues its transactions at a much slower pace.

The results of one such experiment is shown in \subfigref{banking_exp}{banking_ripple}. In this experiment, the TransferMoney transaction program from the banking example (Example \ref{ex:banking}) is used in both streams. The schedule for this experiment is generated using logical time units. The time taken for executing one TransferMoney is assumed to be 250 units for both algorithms.  We choose three quarters of this time, i.e., 187 units as the time for partially re-executing conflicting blocks  in \systemName, based on the results of \subfigref{banking_exp}{banking_window}. The two streams produce TransferMoney transactions at the rate of 251 units and 72,000,000 units respectively. This figure shows not only that the time for processing \systemName transactions is lower, but also that the overall behavior of \systemName is completely different from \NMVCC. Basically, this experiment shows that a longer conflict resolution approach not only affects an individual transaction, but also has a compound effect on the subsequent transactions. The reason behind it is an increase in the probability of having many concurrent transactions, which results in even more conflicts.

\subsection{Overhead of \systemName}

Using \systemName as a generic MVCC algorithm is reasonable, if its overhead in the absence of validation failures, or premature aborts during the execution is not significant. The following subsections show the overhead of \systemName in these cases.

\subsubsection{Conflict-free transactions}

In this section, the overhead imposed by using \systemName in the absence of any conflict is illustrated. This conflict-free execution is highly probable in practical low contention scenarios. We consider two approaches for observing this behavior. The first approach is executing transactions serially, and the second one is running transactions concurrently with no conflicts among them.

As it is illustrated in Figures \subfigrefnum{banking_exp}{banking_window} and \subfigrefnum{trading_exp}{trading_window}, the overhead of \systemName compared to \NMVCC in the serial execution scenario, i.e., window size 1, is under 1\%. In addition, \subfigref{banking_exp}{banking_percentage} shows that the concurrent execution of transactions with no conflicts has less than 1\% overhead using \systemName compared to \NMVCC. This small overhead is mostly related to creating the predicate graph in \systemName, instead of creating a list of predicates as in \NMVCC. It should be noted that applying the compiler optimizations for efficiently compiling the closures bound to predicates in \systemName plays an important role in achieving such a low overhead.

\subsubsection{TPC-C benchmark}

\begin{figure}[htb]
\begin{center}
\leavevmode
\includegraphics[width=\columnwidth]{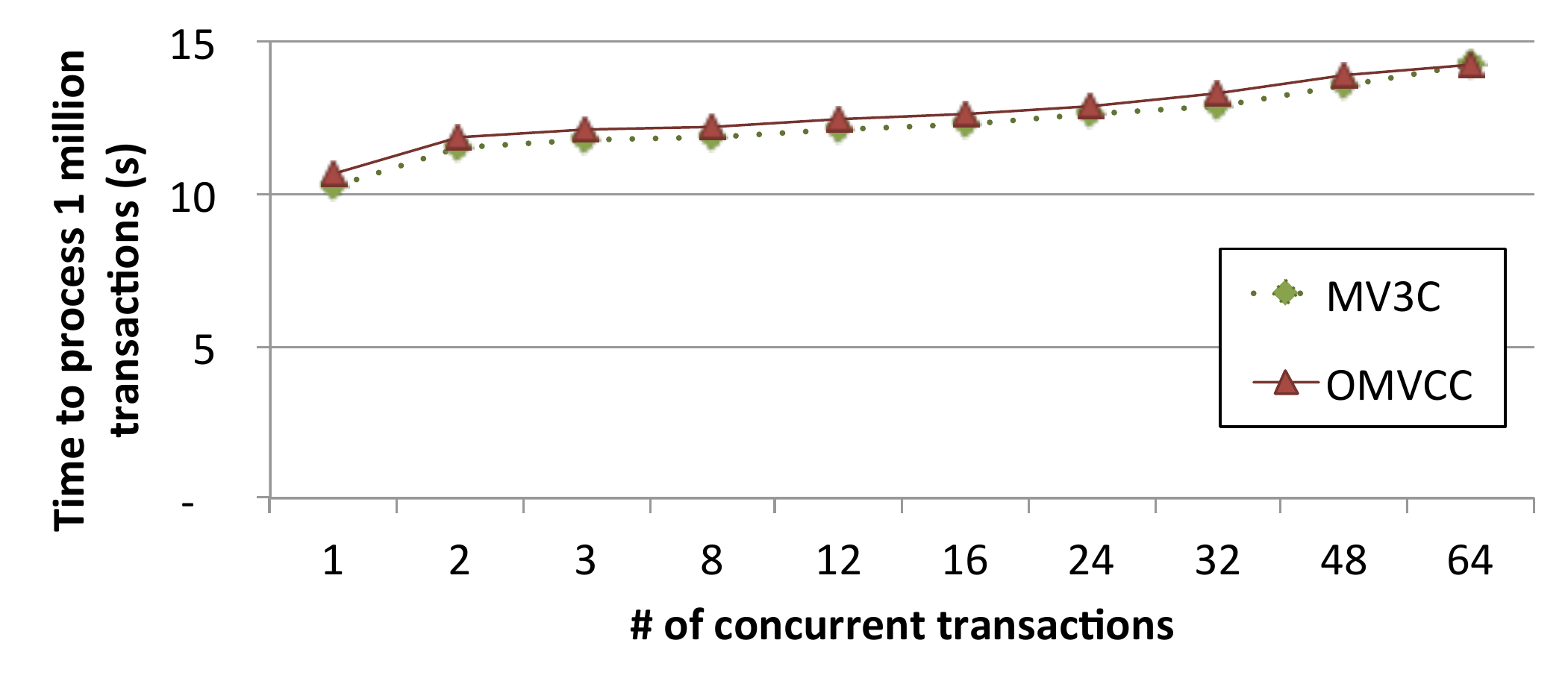}
\end{center}
\vspace*{-0.2cm}
\caption{TPC-C experiment with 1 warehouse}
\label{fig:tpcc}
\end{figure}

The TPC-C benchmark is implemented using both \systemName and \NMVCC with attribute-level predicate validation. The results of this benchmark with scale factor 1, i.e., with one warehouse, are shown in \figref{tpcc}. As it is shown in this figure, the execution times of \systemName and \NMVCC for one million TPC-C transactions are almost the same for different window sizes, even though 92\% percent of TPC-C transactions are not read-only \cite{tozun13}. The reason is that in this scenario, almost all conflicting transactions lead to premature abort during execution, instead of reaching the validation phase and failing it. Consequently, the main advantage of \systemName compared to \NMVCC, i.e., the repair phase, rarely happens. Even though \figref{tpcc} shows that \systemName has almost no advantage in this benchmark, it shows that it has no overhead either.